\documentclass[a4paper,12pt]{article}
 \usepackage{comment}
\usepackage{xr}
\usepackage[margin=1in]{geometry}
\usepackage[utf8]{inputenc}
\usepackage{amsmath}
\usepackage[english]{babel}

\usepackage{natbib}
\usepackage{bbm}
\usepackage{romanbar}
\usepackage{amsthm}
\usepackage{enumerate}
\usepackage{subcaption}
\usepackage{lscape}
\usepackage{graphicx}
\graphicspath{ {figures/} }
\usepackage{color}
\newcommand{\colr}[1]{{\color{red} {#1}}}

\usepackage{booktabs}
\usepackage{longtable}
\usepackage{array}
\usepackage{multirow}
\usepackage{wrapfig}
\usepackage{float}
\usepackage{colortbl}
\usepackage{pdflscape}
\usepackage{tabu}
\usepackage{threeparttable}
\usepackage{threeparttablex}
\usepackage[normalem]{ulem}
\usepackage{makecell}
\usepackage{xcolor}
\newtheorem{thm}{Theorem}[section]

\newtheorem{lem}[thm]{Lemma}

\theoremstyle{definition}

\theoremstyle{remark}

\numberwithin{equation}{section}

\usepackage{accents}

\newcommand{\bol}[1]{\mbox{\boldmath$#1$}}

\newcommand{\eqdist}{\stackrel{d}{=}}
\newcommand{\bSigma}{\mathbf{\Sigma}}

\newcommand{\bOmega}{\mathbf{\Omega}}

\newcommand{\bmu}{\bol{\mu}}

\newcommand{\btheta}{\bol{\theta}}
\newcommand{\bb}{\mathbf{b}}

\newcommand{\bx}{\mathbf{x}}

\newcommand{\by}{\mathbf{y}}

\newcommand{\bC}{\mathbf{C}}

\newcommand{\bB}{\mathbf{B}}

\newcommand{\bX}{\mathbf{X}}
\newcommand{\bY}{\mathbf{Y}}

\newcommand{\bw}{\mathbf{w}}
\newcommand{\hbw}{\mathbf{\hat{w}}}

\newcommand{\bi}{\mathbf{1}}

\newcommand{\bI}{\mathbf{I}}

\newcommand{\tr}{\mbox{tr}}

\newcommand{\bxi}{\boldsymbol{\xi}}
\newcommand{\bD}{\mathbf{D}}

\newcommand{\bV}{\mathbf{V}}

\newcommand{\bS}{\mathbf{S}}

\newcommand{\ones}{\mathbf{1} }

\providecommand{\keywords}[1]
{
\small	
\textbf{\textit{Keywords}:} #1
}
\definecolor{green}{rgb}{0.0, 0.5, 0.0}

\usepackage{authblk}
\title{Two is better than one: Regularized shrinkage of large minimum variance portfolio}
 \author[1]{Taras Bodnar}
 \author[2]{Nestor Parolya\footnote{Corresponding author: Nestor Parolya. E-mail address: n.parolya@tudelft.nl.}}
 \author[1]{Erik Thors{\'e}n}
 \affil[1]{Department of Mathematics, Stockholm University, Roslagsv\"{a}gen 101, SE-10691 Stockholm, Sweden}
 \affil[2]{Department of Applied Mathematics, Delft University of Technology, Mekelweg 4,
 2628 CD Delft, The Netherlands}

\begin{document}
\maketitle

\begin{abstract}
    In this paper we construct a shrinkage estimator of the global minimum variance (GMV) portfolio by a combination of two techniques: Tikhonov regularization and direct shrinkage of portfolio weights. More specifically, we employ a double shrinkage approach, where the covariance matrix and portfolio weights are shrunk simultaneously. The ridge parameter controls the stability of the covariance matrix, while the portfolio shrinkage intensity shrinks the regularized portfolio weights to a predefined target. Both parameters simultaneously minimize with probability one the out-of-sample variance as the number of assets $p$ and the sample size $n$ tend to infinity, while their ratio $p/n$ tends to a constant $c>0$.
This method can also be seen as the optimal combination of the well-established linear shrinkage approach of \cite{lw2004} and the shrinkage of the portfolio weights by \cite{bodnar2018estimation}.  No specific distribution is assumed for the asset returns except of the assumption of finite $4+\varepsilon$ moments. The performance of the double shrinkage estimator is investigated via extensive simulation and empirical studies. The suggested method significantly outperforms its predecessor (without regularization) and the nonlinear shrinkage approach in terms of the out-of-sample variance, Sharpe ratio and other empirical measures in the majority of scenarios. Moreover, it obeys the most stable portfolio weights with uniformly smallest turnover.
\end{abstract}



\newpage

\section{Introduction}

The global minimum variance (GMV) portfolio is the portfolio with the smallest variance among all optimal portfolios, which are the solutions to the mean-variance optimization problem suggested in the seminal paper of Harry Markowitz (see, \cite{markowitz1952}). It is the only optimal portfolio whose weights are solely determined by the covariance matrix of the asset returns and do not depend on the mean vector. Such a property has been recognized to be very important due to the fact that the estimation error in the means are several times as large as the estimation error in the variances and covariances of the asset returns (see, \cite{best1991sensitivity}, \cite{merton1980estimating}).

In the original optimization problem, the GMV portfolio is obtained as the solution of
\begin{equation}\label{gmv_optim_population}
\underset{\bw}{\text{minimize}}\; 
 \bw^\top \bSigma \bw \quad \text{subject to}
\quad \bw^\top \ones = 1 
\end{equation}
and its weights are given by 
\begin{equation}\label{GMV_population}
    \bw_{GMV}=\frac{\bSigma^{-1}\ones}{\ones^\top\bSigma^{-1}\ones}.
\end{equation}
Since the covariance matrix $\bSigma$ is an unknown quantity, the GMV portfolio cannot be constructed by using \eqref{GMV_population}. In \cite{markowitz1959portfolio} the author uses the sample estimator of $\bw_{GMV}$ instead of \eqref{GMV_population}. He considers
\begin{equation}\label{GMV_sample}
    \hbw_{GMV}=\frac{\bS_n^{-1}\ones}{\ones^\top\bS_n^{-1}\ones},
\end{equation}
where $\bS_n$ is the sample estimator of the covariance matrix $\bSigma$ which is given by
\begin{equation}\label{bS}
\bS_n=\frac{1}{n}\left(\bY_n - \bar\by_n\ones^\top \right) \left(\bY_n - \bar\by_n\ones^\top \right)^\top
\quad \text{with} \quad 
\bar\by_n=\frac{1}{n} \bY_n \ones,
\end{equation}
where $\bY_n=[\by_1, ..., \by_n]$ is the $p \times n$ observation matrix and $\by_i$, $i=1,..,n$, is the $p$-dimensional vector of asset returns observed at time $i$. 
As such, the sample GMV portfolio with weights \eqref{GMV_sample} may be considered as the solution of the optimization problem \eqref{gmv_optim_population} where the unknown covariance matrix $\bSigma$ is replaced by $\bS_n$, namely,
\begin{equation}\label{gmv_optim_sample}
\underset{\bw}{\text{minimize}} \;
 \bw^\top \bS_n \bw \quad \text{subject to}
\quad \bw^\top \ones = 1 .
\end{equation}
Recently, several other estimates of the GMV portfolio weights were introduced in the literature (see, e.g., \cite{lw2004}, \cite{golosnoy2007}, \cite{frahm2010}, \cite{demiguel2013}, \cite{lw2017}, \cite{bodnar2018estimation} and \cite{lw2020} with references therein). All of these methods either shrink the covariance matrix and use it for the estimation of GMV portfolio, or they shrink the portfolios weights directly. To the best of our knowledge, none of the known approaches combines both procedures into one.

As previously stated, there are many ways to cope with estimation uncertainty or the fact that using $\bS_n$ instead of $\bSigma$ may produce a very noisy estimator of the portfolio weights. Our approach relies on two distinct features. First, the linear shrinkage estimator from \citet{bodnar2018estimation} has proven to provide good results in terms of the out-of-sample variance and to be robust for the large dimensional portfolios. It does not, however, reduce the size of the positions or the variance (as measures by turnover) of the portfolio weights, (see, e.g., \citet{bodnar2021dynamic}). This leads us to the second feature. The single source of uncertainty in the GMV portfolio is $\bS_n$. If we can stabilize or decrease the variance of the sample covariance matrix we may decrease the variance of the weights. Our aim is therefore to shrink the sample covariance matrix as well. We apply the Tikhonov regularization (see, e.g., \cite{tikhonov1995numerical})  to the optimization problem \eqref{gmv_optim_sample}, namely 
\begin{equation}\label{gmv_optim_Tikh}
\underset{\bw}{\text{minimize}} \;
 \bw^\top \bS_n \bw + \eta \bw^\top \bw \quad \text{subject to}
\quad \bw^\top \ones = 1 ,
\end{equation}
where $\eta$ is the regularization parameter. Similar approaches is used in the regression analysis, where the ridge regression uses the Tikhonov regularization to stabilize the least-squared estimator of the coefficients of the regression line (cf., \cite{golub1999tikhonov}). The solution of \eqref{gmv_optim_Tikh} is given by
\begin{equation}\label{GMV_Tikh}
    \bw_{S; \lambda}=\frac{(\bS_n+\eta \bI)^{-1}\ones}{\ones^\top(\bS_n+\eta \bI)^{-1}\ones}.
\end{equation}
Without loss of generality we set $\eta=\frac{1}{\lambda}-1$ where $\lambda \in (0,1]$. 
Using this representation and $$\bS_\lambda = \lambda \bS_n + (1-\lambda)\bI$$ instead of $\bSigma$ in \eqref{GMV_population} results in the same solution. However, the corner solutions are more easily understood as $\bS_\lambda$ is a simple convex combination. If $\lambda \rightarrow 0$ then we put all our beliefs in the diagonal matrix $\bI$, whereas if $\lambda \rightarrow 1$ then all beliefs are placed in $\bS_n$ (see, \cite{lw2004}).

Finally, combining the linear shrinkage estimator from \citet{bodnar2018estimation} we shrink the already regularized GMV portfolio weights as follows
\begin{equation}\label{eqn:def_weights_int}
    \hbw_{Sh;\lambda, \psi} = \psi\hbw_{S; \lambda} + (1-\psi)\bb,
\end{equation}
where $\psi$ is the shrinkage intensity towards the target portfolio $\bb$ with $\bb^\top\bi=1$. 
This approach allows us to shrink the sample covariance matrix to decrease the variance and further decrease it by shrinking the weights themselves. It also gives a way for an investor to highlight stocks of he/she likes with the target portfolio $\bb$. In many cases a naive portfolio $\bb=\frac{1}{p}\bi$ is a good choice but, in general, any deterministic target, which reflects the investment beliefs, is possible.

A common approach to determine shrinkage intensities is to use cross-validation (see, e.g., \cite{tong2018} and \cite{boileau2021}). That is, one aims to find the parameters $\lambda$ and $\psi$ such that some loss (or metric) $L(\lambda,\psi)$ is minimized (or maximized). Since we are working with the GMV portfolio, the most natural choice of loss is the out-of-sample variance given by
$$ 
L(\lambda,\psi)=\hbw_{Sh;\lambda, \psi}^\top \bSigma \hbw_{Sh;\lambda, \psi},$$
Howver, since $\bSigma$ is not known we need to estimate it. Using K-fold cross-validation one partitions the data into train and validation sets. Using these sets one can perform a grid-search and use the validation set to estimate $\bSigma$. The empirical out-of-sample variance is thereafter aggregated over fold to determine the best pair of shrinkage coefficients. However, using grid-search and cross-validation for such a problem introduces several obstacles. We have established that the sample covariance is a noisy estimate and validation sets are usually smaller than training sets. That naturally creates a more volatile estimator. Furthermore, it is not clear how big should be a grid. The approach we develop needs neither resampling methods nor grid search but instead relies on methods from random matrix theory (see, \cite{bai2010spectral}). We develop a bona-fide type loss function which consistently estimates the true loss function in the high dimensional setting. The problem reduces to a simple univariate nonlinear optimization problem, which can be efficiently solved numerically.

The rest of the paper is organized as follows. In Section \ref{sec:asym_loss} the asymptotic properties of the out-of-sample variance are investigated in the high dimensional setting, while Section \ref{sec:BF} presents a bona fide estimator of the asymptotic loss, which is the n used to find the optimal values of the two shrinkage intensities. The results of en extensive simulation study and of empirical applications are provided in Section \ref{sec:num}, while Section \ref{sec:sum} summarizes the obtained findings. The mathematical derivations are moved to the appendix (Section \ref{sec:app}).

\section{Out-of-sample variance and shrinkage estimation}\label{sec:asym_loss}
Let $\bX_n$ be a matrix of size $p \times n$ where its elements $\{x_{ij}\}_{ij}$ are independent and identically distributed (i.i.d.) real random variables with zero mean, unit variance and finite $4+\epsilon$ moment for some $\epsilon>0$. Assume that we observe the matrix $\bY_n$ according to the stochastic model
\begin{align}\label{eqn:obs}
    \bY_n & \eqdist \bmu \ones^\top +  \bSigma^{\frac{1}{2}} \bX_n
\end{align}
where $\bSigma$ is a positive definite matrix of size $p \times p$ with a bounded spectral norm (its minimum and maximum eigenvalues are uniformly bounded in $p$ from zero and infinity, respectively)\footnote{\scriptsize In fact the obtained results can be generalized to the case with finite number of unbounded largest eigenvalues, which would make the proofs more lengthy. Moreover,  one can show that this assumption is only needed in case of centered sample covariance matrix, i.e., unknown mean vector. In case $\bmu$ is known, the boundedness of eigenvalues may be ignored due to normalization presented further in \eqref{normalization}. More details could be deduced from the proofs of the main theorems.}. The model belongs to the location-scale family but includes many skew or bi-modal families as well. Our aim is to estimate the shrinkage intensities $\lambda$, $\psi$ from the following normalized optimization problem
\begin{equation}\label{normalization}
    \min_{\lambda, \psi} \frac{\hbw_{Sh;\lambda, \psi}^\top\bSigma \hbw_{Sh;\lambda, \psi}}{\bb^\top \bSigma \bb}.
\end{equation}
The normalization is merely a technicality. The out-of-sample variance, or the loss function $L(\lambda, \psi)=\hbw_{Sh;\lambda, \psi}^\top\bSigma \hbw_{Sh;\lambda, \psi}$, can be further simplified to
\begin{align}\label{Li-lambda}
    L(\lambda, \psi)
    = & \left(\psi\hbw_{S; \lambda} + (1-\psi) \bb \right)^\top\bSigma
    \left(\psi\hbw_{S; \lambda} + (1-\psi) \bb \right) \nonumber \\
    = & 
    \left(\bb-\hbw_{S;\lambda}\right)^\top \bSigma\left(\bb-\hbw_{S;\lambda}\right)  
    \left(
        \psi- \frac{\bb^\top \bSigma \left(\bb-\hbw_{S;\lambda}\right)} {\left(\bb-\hbw_{S;\lambda}\right)^\top \bSigma\left(\bb-\hbw_{S;\lambda}\right)}
    \right)^2 
    \nonumber\\
    &-\frac{\left(\bb^\top \bSigma \left(\bb-\hbw_{S;\lambda}\right)\right)^2}
    {\left(\bb-\hbw_{S;\lambda}\right)^\top \bSigma\left(\bb-\hbw_{S;\lambda}\right)}
    +\bb^\top\bSigma \bb.
 \end{align}
For given $\lambda$, the out-of-sample variance is minimized with respect to $\psi$ at
\begin{equation}\label{alp_ni_star-lambda}
    \psi_n^*(\lambda) = 
    \frac{
        \hbw_{S;\lambda}^\top \bSigma \left(\bb - \hbw_{S;\lambda}\right)
    }{
        \left(\bb - \hbw_{S;\lambda}\right)^\top \bSigma\left(\bb-\hbw_{S;\lambda}\right)
    },
\end{equation}
while the optimal $\lambda$ is found by maximizing the normalized second summand in \eqref{Li-lambda}, i.e.,
\begin{equation}\label{Li2-lambda}
    L_{n;2}(\lambda)=\frac{1}{\bb^\top\bSigma \bb}
        \frac{\left(\bb^\top \bSigma \left(\bb-\hbw_{S;\lambda}\right)\right)^2}
        {\left(\bb-\hbw_{S;\lambda}\right)^\top \bSigma\left(\bb-\hbw_{S;\lambda}\right)}.
\end{equation}

In order to find the value $\lambda^*$ together with $\psi^*(\lambda)$, which minimize the loss function, we proceed in three steps. First, we find the deterministic equivalent to $L_{n;2}(\lambda)$, and estimate it consistently in the second step. Finally, we minimize the obtained consistent estimator in the last step.


\begin{thm}\label{th1-lam}
Let $\bY_n$ possess the stochastic representation as in \eqref{eqn:obs}. Assume that the relative loss of the target portfolio expressed as
\begin{equation}\label{eq:rel_loss}
    L_{\bb}=\frac{\bb^\top\bSigma \bb - \frac{1}{\ones^\top\bSigma^{-1}\ones}}{\frac{1}{\ones^\top\bSigma^{-1}\ones}}=\bb^\top\bSigma \bb \ones^\top\bSigma^{-1}\ones-1
\end{equation}
is uniformly bounded in $p$. Then it holds that
\begin{enumerate}[(i)]
    \item 
    \begin{equation}\label{Li2-lambda-as}
        \left|L_{n;2}(\lambda)-L_{2}(\lambda)\right| \stackrel{a.s.}{\rightarrow} 0 
    \end{equation}
for $p/n \to c \in (0,\infty)$ as $n \to \infty$ with
\begin{equation}\label{Li2-lambda-star}
    L_{2}(\lambda)=
    \frac{
        \left(1-\frac{1}{\bb^\top \bSigma \bb} \frac{
        \bb^\top \bSigma \bOmega_{\lambda}^{-1}\ones}{\ones^\top \bOmega_{\lambda}^{-1}\ones} \right)^2
        }{
        1
        -\frac{2}{\bb^\top \bSigma \bb}\frac{\bb^\top \bSigma \bOmega_{\lambda}^{-1}\ones}{\ones^\top \bOmega_{\lambda}^{-1}\ones}
        +\frac{1}{\bb^\top \bSigma \bb} \frac{(1-v_2^\prime(\eta, 0))\ones^\top  \bOmega_{\lambda}^{-1} \bSigma \bOmega_{\lambda}^{-1} \ones}{\left(\ones^\top \bOmega_{\lambda}^{-1}\ones\right)^2} } 
\end{equation}
    \item
    \begin{equation}\label{psi-as}
  \left|\psi_n^*(\lambda)-\psi^*(\lambda)\right| \stackrel{a.s.}{\rightarrow} 0 
\end{equation}
for $p/n \to c \in (0,\infty)$ as $n \to \infty$ with
\begin{equation}\label{psi-star}
\psi^*(\lambda)= 
    \frac{
        \frac{
            1
        }{
            \bb^\top \bSigma \bb
        }\frac{
            \bb^\top \bSigma \bOmega_{\lambda}^{-1}\ones
            }{
            \ones^\top \bOmega_{\lambda}^{-1}\ones
            } 
        -\frac{1}{\bb^\top \bSigma \bb} 
        \frac{
                (1-v_2^\prime(\eta, 0))\ones^\top  \bOmega_{\lambda}^{-1} \bSigma \bOmega_{\lambda}^{-1} \ones
            }{
                \left(\ones^\top \bOmega_{\lambda}^{-1}\ones\right)^2
            }
        }{
        1
        -\frac{2}{\bb^\top \bSigma \bb} \frac{\bb^\top \bSigma \bOmega_{\lambda}^{-1}\ones}{\ones^\top \bOmega_{\lambda}^{-1}\ones} 
        +\frac{1}{\bb^\top \bSigma \bb} \frac{(1-v_2^\prime(\eta, 0))\ones^\top  \bOmega_{\lambda}^{-1} \bSigma \bOmega_{\lambda}^{-1} \ones}{\left(\ones^\top \bOmega_{\lambda}^{-1}\ones\right)^2}},
\end{equation}
\end{enumerate}
where
\begin{equation}\label{Omega-lam}
\eta= \frac{1}{\lambda}-1, \quad \bOmega_{\lambda}= v\left(\eta,0\right)\lambda\bSigma + (1-\lambda)\bI,    
\end{equation}
$v(\eta, 0)$ is the solution of the following equation
\begin{equation}\label{v_eta}
v(\eta, 0) = 1-c\left(1-\frac{\eta}{p}\tr\left(\left(v(\eta,0)\bSigma+\eta\bI \right)^{-1}\right) \right),
\end{equation}
and $v_2^\prime(\eta, 0)$ is computed by
\begin{equation}\label{v2_eta}
v_2^\prime(\eta, 0)  =1-\frac{1}{v(\eta,0)}+\eta\frac{v_1^\prime(\eta,0)}{v(\eta,0)^2}.
\end{equation}
with
\begin{equation}\label{v1_eta}
v_1^\prime(\eta, 0)  = v(\eta,0)\frac{c\frac{1}{p}\text{tr}\left((v(\eta,0)\bSigma+\eta\bI)^{-1}\right)-c\eta \frac{1}{p}\text{tr}\left((v(\eta,0)\bSigma+\eta\bI)^{-2}\right)}
{1-c+2c\eta\frac{1}{p}\text{tr}\left((v(\eta,0)\bSigma+\eta\bI)^{-1}\right)-
c\eta^2\frac{1}{p}\text{tr}\left((v(\eta,0)\bSigma+\eta\bI)^{-2}\right)}.
\end{equation}
\end{thm}
The proof of Theorem \ref{th1-lam} can be found in the appendix. Theorem \ref{th1-lam} provides the deterministic equivalents for the loss function $L_{n;2}(\lambda)$ and optimal shrinkage intensity $\psi_n^*(\lambda)$. This is, however, not enough because both depend on the unknown parameters, i.e., $\bSigma$ itself and $v(\eta, 0)$. Fortunately, the consistent bona fide\footnote{With ``bona fide'' we understand a concept of purely data-driven estimators, which do not depend on the unknown quantities. Thus, they are ready to be used in practice without any modifications.} estimators for both deterministic equivalents $L_{2}(\lambda)$ and $\psi^*(\lambda)$ can be found in case of increasing dimension asymptotics.

\section{Bona fide estimation}\label{sec:BF}

In this section, we construct bona fide consistent estimators for $L_{2}(\lambda)$ and $\psi^*(\lambda)$ in the high-dimensional asymptotic setting.
First, in Theorem \ref{th2-lam} we derive the consistent estimators for $v(\eta,0)$, $v_1^\prime(\eta,0)$, and  $v_2^\prime(\eta,0)$. The proof of Theorem \ref{th2-lam} is given in the appendix.

\begin{thm}\label{th2-lam}
Let $\bY_n$ possess the stochastic representation as in \eqref{eqn:obs}. 
Then it holds that
\begin{eqnarray}
\left|\hat{v}(\eta, 0)-v(\eta, 0)\right| &\stackrel{a.s}{\rightarrow}& 0, \label{lem:v_est-eq1}\\
\left|\hat{v}_1^\prime(\eta, 0)-v_1^\prime(\eta, 0)\right| &\stackrel{a.s}{\rightarrow}& 0,\label{lem:v_est-eq2}\\
\left|\hat{v}_2^\prime(\eta, 0)-v_2^\prime(\eta, 0)\right| &\stackrel{a.s}{\rightarrow}& 0,\label{lem:v_est-eq3}
\end{eqnarray}
for $p/n\rightarrow c \in (0,\infty)$ as $n\rightarrow\infty$ with 
\begin{eqnarray*}
\hat{v}(\eta, 0) &=& 1-c\left(1-\eta\frac{1}{p}\tr\left(\left(\bS_n+\eta\bI \right)^{-1}\right) \right),\\
\hat{v}_1^\prime(\eta, 0) 
&=&\hat{v}(\eta,0) c \left(\frac{1}{p}\tr\left(\left(\bS_n+\eta\bI \right)^{-1}\right)-\eta \frac{1}{p}\tr\left(\left(\bS_n+\eta\bI \right)^{-2}\right)\right),\\
\hat{v}_2^\prime(\eta, 0) &=&
1-\frac{1}{\hat{v}(\eta,0)}+\eta\frac{\hat{v}_1^\prime(\eta,0)}
{\hat{v}(\eta,0)^2}.
\end{eqnarray*}
\end{thm}

Theorem \ref{th3-lam} provides the consistent estimators for the building blocks used in the construction of the consistent estimators for $L_{2}(\lambda)$ and $\psi^*(\lambda)$. The proof of Theorem \ref{th3-lam} is presented in the appendix.

\begin{thm}\label{th3-lam}
Let $\bY_n$ possess the stochastic representation as in \eqref{eqn:obs}. Assume that the relative loss of the target portfolio given in \eqref{eq:rel_loss}
is uniformly bounded in $p$. Let $\frac{\bb^\top \bSigma^{-1}\bb}{\bb^\top \bSigma\bb}$ be uniformly bounded in $p$. Then it holds that
\begin{flalign}
& \left|\frac{\bb^\top\bS\bb}{\bb^\top\bSigma\bb} -1\right|\stackrel{a.s}{\rightarrow} 0, \label{eq:lem58_0}\\
& \left|\frac{\ones^\top\bS_{\lambda}^{-1}\ones}{\ones^\top\bSigma^{-1}\ones} - \frac{\ones^\top\bOmega_{\lambda}^{-1}\ones}{\ones^\top\bSigma^{-1}\ones}
        \right|  \stackrel{a.s}{\rightarrow} 0, \label{eq:lem58_1}\\
&\left|\frac{\lambda^{-1}}{ \hat{v}(\eta, 0)}\frac{1-(1-\lambda) \bb^\top \bS_\lambda^{-1} \ones}{\sqrt{\bb^\top\bSigma\bb\ones^\top\bSigma^{-1}\ones}} - \frac{\bb^\top \bSigma \bOmega_\lambda^{-1} \ones }{\sqrt{\bb^\top\bSigma\bb\ones^\top\bSigma^{-1}\ones}}
        \right|   \stackrel{a.s}{\rightarrow} 0, \label{eq:lem58_2} \\
&\Bigg|\frac{1}{\lambda \hat{v}(\eta, 0)} \frac{\ones^\top\bS_{\lambda}^{-1}\ones}{\ones^\top\bSigma^{-1}\ones}  -   \frac{1-\lambda}{\lambda \hat{v}(\eta, 0)}\frac{\ones^\top \bS_{\lambda}^{-2}\ones  - \lambda^{-1}
                \frac{\hat{v}_1^\prime(\eta,0)}{\hat{v}(\eta,0)}\ones^\top \bS_{\lambda}^{-1} \ones
            }{\ones^\top\bSigma^{-1}\ones\left(                1-\frac{\hat{v}_1^\prime(\eta,0)}{\hat{v}(\eta,0)}
            (\frac{1}{\lambda}-1)\right)
            }
        - \frac{\ones^\top\bOmega^{-1}_{\lambda}\bSigma\bOmega^{-1}_{\lambda}\ones}{\ones^\top\bSigma^{-1}\ones}
        \Bigg| \stackrel{a.s}{\rightarrow} 0\label{eq:lem58_3}
    \end{flalign}
for $p/n\rightarrow c \in (0,\infty)$ as $n\rightarrow\infty$ with $\eta=1/\lambda-1$.
\end{thm}

Let
\begin{equation}\label{d1}
d_1 (\eta)=  \frac{\lambda^{-1}}{ \hat{v}(\eta, 0)}\left(1-(1-\lambda) \bb^\top \bS_\lambda^{-1} \ones\right)  
\end{equation}
and
\begin{equation}\label{d2}
d_2(\eta)= \frac{1}{\lambda \hat{v}(\eta, 0)} \ones^\top\bS_{\lambda}^{-1}\ones  -   \frac{1-\lambda}{\lambda \hat{v}(\eta, 0)}\frac{\ones^\top \bS_{\lambda}^{-2}\ones  - \lambda^{-1}                \frac{\hat{v}_1^\prime(\eta,0)}{\hat{v}(\eta,0)}\ones^\top \bS_{\lambda}^{-1} \ones}
{1-\frac{\hat{v}_1^\prime(\eta,0)}{\hat{v}(\eta,0)}(\frac{1}{\lambda}-1)}.    
\end{equation}
The application of the results derived in Theorems \ref{th2-lam} and \ref{th3-lam} leads to a consistent bona fide estimator for $L_{2}(\lambda)$ and $\psi^*(\lambda)$ presented in Theorem \ref{th4-lam}.

\begin{thm}\label{th4-lam}
Let $\bY_n$ possess the stochastic representation as in \eqref{eqn:obs}. Assume that the relative loss of the target portfolio given in \eqref{eq:rel_loss} is uniformly bounded in $p$. Let $\frac{\bb^\top \bSigma^{-1}\bb}{\bb^\top \bSigma\bb}$ be uniformly bounded in $p$.
 Then it holds that
\begin{enumerate}[(i)]
    \item 
    \begin{equation}\label{hLi2-lambda-as}
        \left|\hat{L}_{n;2}(\lambda)-L_{2}(\lambda)\right| \stackrel{a.s.}{\rightarrow} 0 
    \end{equation}
for $p/n \to c \in (0,\infty)$ as $n \to \infty$ with
\begin{equation}\label{hLi2-lambda-star}
    \hat{L}_{n;2}(\lambda)=
    \frac{
        \left(1-\frac{1}{\bb^\top \bS \bb} \frac{d_1(\eta)}{\ones^\top \bS_{\lambda}^{-1}\ones} \right)^2
        }{
        1-\frac{2}{\bb^\top \bS \bb}\frac{d_1(\eta)}{\ones^\top \bS_{\lambda}^{-1}\ones}
        +\frac{1}{\bb^\top \bS \bb} \frac{(1-\hat{v}_2^\prime(\eta, 0))d_2(\eta)}{\left(\ones^\top \bS_{\lambda}^{-1}\ones\right)^2} } ,
\end{equation}
    \item
    \begin{equation}\label{hpsi-as}
  \left|\hat{\psi}_n^*(\lambda)-\psi^*(\lambda)\right| \stackrel{a.s.}{\rightarrow} 0 
\end{equation}
for $p/n \to c \in (0,\infty)$ as $n \to \infty$ with
\begin{equation}\label{hpsi-star}
\hat{\psi}_n^*(\lambda)= 
    \frac{
        \frac{1}{\bb^\top \bS \bb
        }\frac{d_1(\eta)}{ \ones^\top \bS_{\lambda}^{-1}\ones} 
        -\frac{1}{\bb^\top \bS \bb} 
        \frac{(1-\hat{v}_2^\prime(\eta, 0))d_2(\eta)
            }{\left(\ones^\top \bS_{\lambda}^{-1}\ones\right)^2}}
        {1-\frac{2}{\bb^\top \bS \bb} \frac{d_1(\eta)}{\ones^\top \bS_{\lambda}^{-1}\ones} 
        +\frac{1}{\bb^\top \bS \bb} \frac{(1-\hat{v}_2^\prime(\eta, 0))d_2(\eta)}{\left(\ones^\top \bS_{\lambda}^{-1}\ones\right)^2}},
\end{equation}
\end{enumerate}
where $\eta= {1}/{\lambda}-1$, $\hat{v}_2^\prime(\eta, 0)$ is provided in Theorem \ref{th2-lam}, $d_1(\eta)$ and $d_2(\eta)$ are given in \eqref{d1} and \eqref{d2}, respectively.
\end{thm}

The three loss functions (bona fide, oracle and true) are illustrated in Figure \ref{fig:loss_figures} for two different values of $p$. Under oracle loss function we understand the asymptotic equivalent of $L_{n;2}$, namely $L_2$. When $p$ is equal to 150 the differences between the functions look, at least graphically, very small. The relative difference it at most $20\%$ and the optimal $\lambda$'s are extremely close to each other. We still need to face the fact that the out-of-sample variance is slightly over-estimated, even though the bona-fide estimator will be asymptotically valid. When $p$ is equal to $450$ and $c$ is greater than one we observe a slightly different picture. For $p>n$ the bona fide loss function is not necessarily concave. This is due to the fact that as $\lambda$ approaches 1, and $c$ is greater than one, $\lambda \bS_n + (1-\lambda)\bI$ becomes closer and closer to a singular matrix. Thus, the eigenvalues of the inverse of the shrunk sample covariance matrix explode. This issue could be repaired using a different type of ridge regularization mentioned in \cite[formula (2.33)]{bop2021}, where Moore-Penrose inverse for $c>1$ can be employed for $\lambda\to1$. This interesting observation is left for the future investigations.

\begin{figure}
\begin{subfigure}{0.45\textwidth}
    \includegraphics[width=\textwidth]{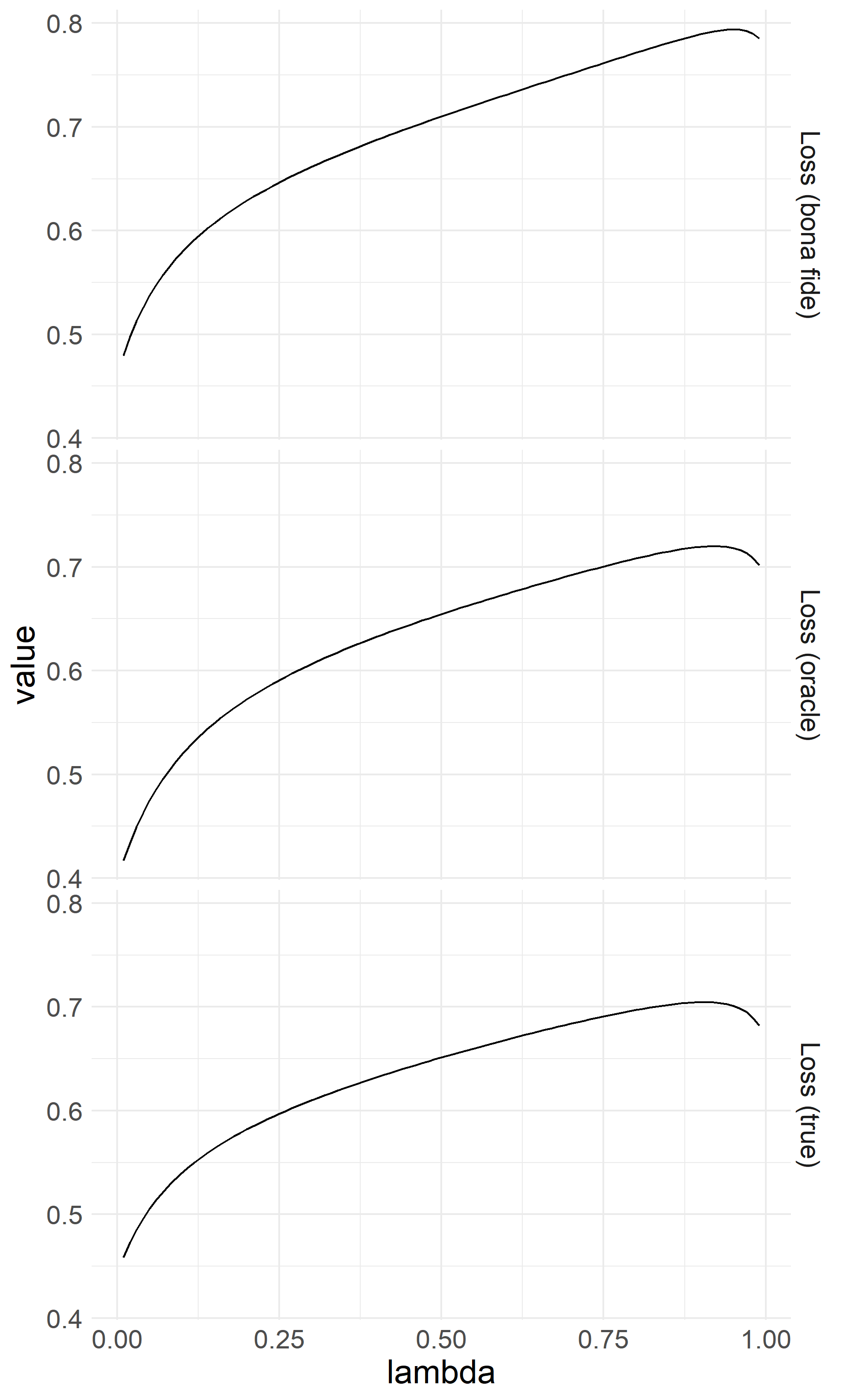}
    \caption{p=150, n=300}
    \label{fig:loss_ex1}
\end{subfigure}
\hfill
\begin{subfigure}{0.45\textwidth}
    \includegraphics[width=\textwidth]{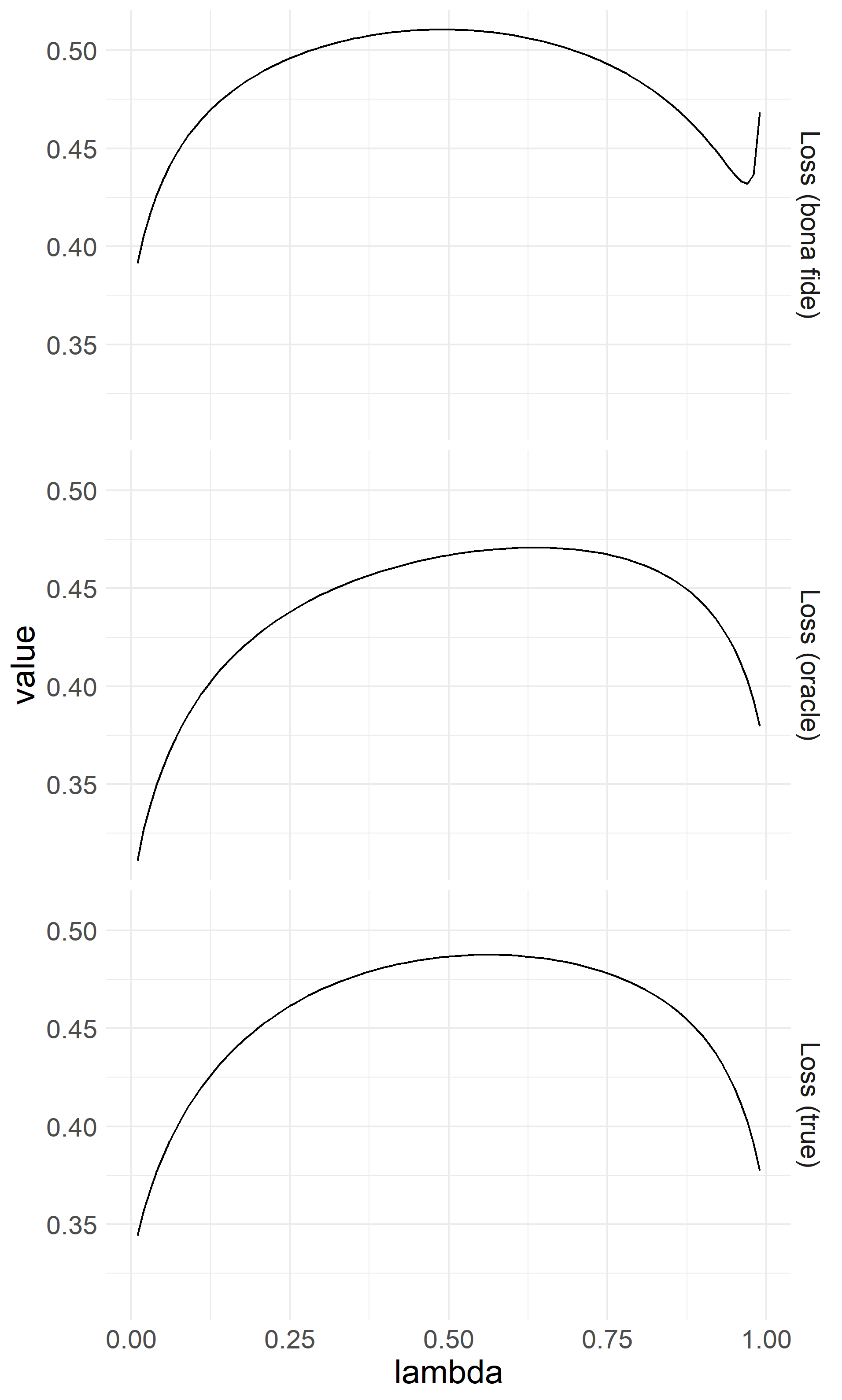}
    \caption{p=450, n=300}
    \label{fig:loss_ex2}
\end{subfigure}
\caption{The loss functions from Theorem \ref{th1-lam}-\ref{th4-lam} illustrated over different values of $\lambda \in (0,1)$. The data were simulated from a t-distribution with 5 degrees of freedom and used the equally weighted portfolio as a target.\label{fig:loss_figures}}
\end{figure}

\section{Numerical study}\label{sec:num}
In this section we will conduct a simulation study to assess the finite sample properties of the suggested double shrinkage estimator and to compare its behaviour with existent approaches. Due to the asymptotic nature of our procedure we first devote some attention to the finite sample properties of the suggested estimator under different data-generating processes. We end this section with an empirical application of the methods on the assets from S\&P 500.
\subsection{Setup of the simulation study}
In the simulation study we will use the following four different stochastic models for the data-generating process:
\begin{description}
    \item[Scenario 1: $t$-distribution]
    The elements of $\bx_t$ are drawn independently from $t$-distribution with $5$ degrees of freedom, i.e., $x_{tj}\sim t(5)$ for $j=1,...,p$, while $\by_t$ is constructed according to \eqref{eqn:obs}.
    \item[Scenario 2: CAPM] 
    The vector of asset returns $\by_t$ is generated according to the CAPM (Capital Asset Pricing Model), i.e.,
    $$\by_t = \bmu + \bol{\beta} z_t+ \bSigma^{1/2}\bx_t,$$
    with independently distributed $z_t \sim N(0, 1)$ and $\bx_{t} \sim N_p(\mathbf{0}, \mathbf{I})$. The elements of vector $\bol{\beta}$ are drawn from the uniform distribution, that is $\beta_i \sim U(-1,1)$ for $i=1,...,p$.
    \item[Scenario 3: CCC-GARCH model of \cite{bollerslev1990modelling}] The asset returns are simulated according to 
    $$\by_t | \bSigma_t \sim N_p(\bmu, \bSigma_t)$$
    where the conditional covariance matrix is specified by 
    $$\bSigma_t = \bD_t^{1/2} \bC \bD_t^{1/2}
    \quad \text{with} \quad \bD_t = \operatorname{diag}(h_{1,t}, h_{2,t}, ..., h_{p,t}),$$
    where
    $$
    h_{j,t} = \alpha_{j,0} + \alpha_{j,1} (\by_{j, t-1} - \bmu_j)^2 + \beta_{j,1} h_{j, t-1}, \text{ for } j=1,2,...,p, \text{ and } t=1,2,...,n_i,~i=1,...,T.
    $$
    The coefficients of the CCC model are sampled according to $\alpha_{j,1} \sim U(0,0.1)$ and $\beta_{j,1} \sim U(0.6,0.7)$ which implies that the stationarity conditions, $\alpha_{j,1} + \beta_{j,1} < 1$, are always fulfilled. The constant correlation matrix $\bC$ is induced by $\bSigma$. The intercept $\alpha_{j,0}$ is chosen such that the unconditional covariance matrix is equal to $\bSigma$.
    \item[Scenario 4: VARMA model] The vector of asset returns $\by_t$ is simulated according to  
    $$
    \by_{t} = \bmu + \mathbf{\Gamma} (\by_{t-1}-\bmu) + \bSigma^{1/2}\bx_t
    \quad \text{with} \quad 
    \bx_t\sim N_p(\mathbf{0},\mathbf{I})
    $$
    for $t=1,...,n+m$, where $\mathbf{\Gamma} = \operatorname{diag}(\gamma_1, \gamma_2,..., \gamma_p)$ with $\gamma_i \sim U(-0.9, 0.9)$ for $i=1,...,p$. Note that in the case of the VAR model, the covariance matrix of $\by_t$ is computed as $\text{vec}(\mathbbm{V}ar(\by))=(\bI-\mathbf{\Gamma}\otimes \mathbf{\Gamma})^{-1}\text{vec}(\bSigma)$ where $\text{vec}$ denotes the vec operator. This matrix is thereafter used in the computation of the limiting objects.
\end{description} 

We will repeat each scenario 1000 times for a number of configurations where the concentration ratio $c$ will range from $0.25$ to $2.7$ and $n=100,200,300,400$. The portfolios contain at most $1080$ assets which implies that we are estimating close to $600 000$ parameters as well as the two shrinkage coefficients.
The parameters of the model are simulated in the following manner.
The elements of the mean vector $\bmu$ are simulated from a uniform distribution with $\mu_i \sim U(-0.1, 0.1)$. To simulate the covariance matrix we make use of the function \textbf{RandCovMat} from the HDShOP package (see, \citet{HDShOP}).




\subsection{Comparison to benchmark strategies}
In this section we will investigate the performance of five different methods. We will consider the following type of portfolios
\begin{enumerate}
    \item The portfolio allocation problem obtained from Theorem \ref{th4-lam}, which we will abbreviate "Double".
    \item The linear shrinkage estimator of the GMV portfolio weights from \citet{bodnar2018estimation}, which we will abbreviate "BPS"
    \item The linear shrinkage estimator of the GMV portfolio weights from \citet{frahm2010}, which we will abbreviate "FM". This portfolio can be constructed only for $c<1$ following the approach suggested in \citet{frahm2010}.
    \item The nonlinear shrinkage estimator of the covariance matrix from \citet{ledoit2020analytical} which is used in the GMV portfolio as a replacement for $\bSigma$. We will abbreviate this portfolio strategy as "LW2020".
    \item The traditional GMV portfolio which we will abbreviate "Traditional". Whenever $c>1$ we will use the Moore-Penrose inverse of $\bS_n$ to compute the sample weights of the GMV portfolio.
\end{enumerate}

Notice that the first three types of portfolios can take many target portfolios. We therefore include two target portfolios as well as other benchmarks in the forthcoming comparisons. These are the equally weighted portfolio and an equal correlation target. The first is deterministic and does not depend on data, which is in line with what Theorem \ref{th4-lam} assumes. The second target portfolio depends on data. It assumes that the all assets share the same correlation but have different volatility. 
For each scenario we will display the relative loss $V_\bw / V_{GMV} -1$, where $V_\bw=\bw^\top\bSigma\bw$. A value close to zero indicates a good estimate of the in-sample loss. 

In Figure \ref{fig:lossTdist} we can see the results of the simulations under scenario 1. Each color represent a "type" of portfolio while the linetype highlights what type of target BPS, Double and FM use. For small $c$ the loss is not significantly different between the different types. However, as $c$ becomes larger, the results diverge from each other. Regardless of $n$ the largest loss is provided by Traditional portfolio. The Traditional estimator is famously bad in higher dimensions, see e.g. \citet{bodnar2018estimation}. The FM portfolio is only defined for $c<1$ so the loss for this method is not presented thereafter. The third best method is BPS using equal correlation and equally weighted as target portfolios. The uncertainty from the target portfolio does not seem to impact the loss a lot. It achieves the same amount of loss as the equally weighted target. The best performing portfolios are the Double and LW2020. For smaller $n$ the difference is more pronounced, LW2020 provides the smallest loss. However, as $n$ increases the Double portfolio with the equally weighted target provides quite similar performance. The LW2020 portfolio is 3\% better (lower loss) in comparison to the Double portfolio when $n=400$.

\begin{figure}
    \centering
    \includegraphics[width=\textwidth]{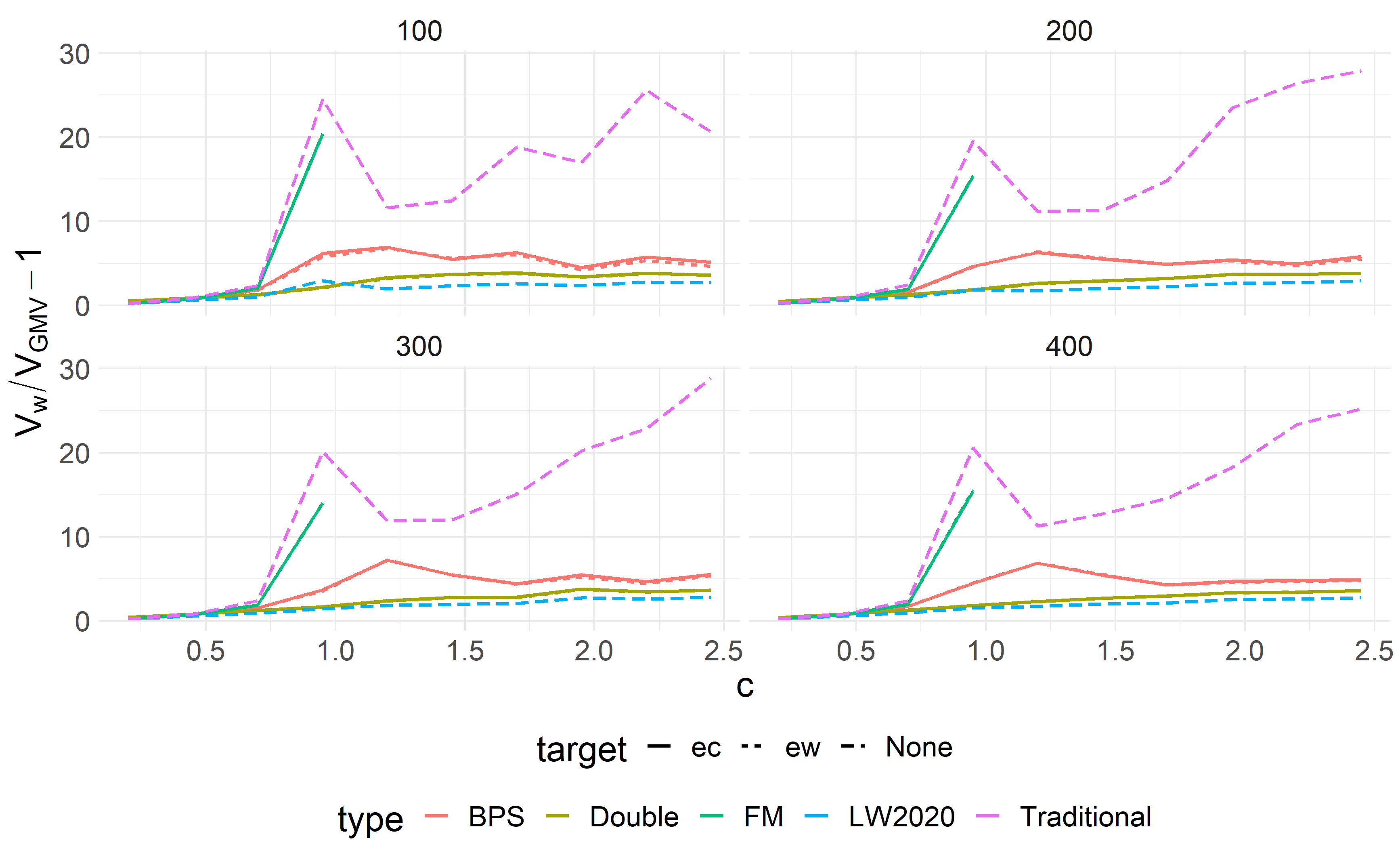}
    \caption{Relative loss $V_{\bw} / V_{GMV} - 1$ computed for several estimator of the GMV portfolio weights under scenario 1. Notice that some methods use different targets but are of the same type.}
    \label{fig:lossTdist}
\end{figure}

In Figure \ref{fig:lossCAPMdist} we can see the results of the simulation study conducted under scenario 2. As with scenario 1, we can see the same type of ordering. However, in scenario 2 there is a pronounced increase for the loss of LW2020 when $n=100$ around $c=1$. The Double portfolios do not seem to suffer from that issue. In all other cases the loss is similar with the smallest difference between Double and LW2020 equal to 3.1\%. The difference is not large since the inverse covariance matrix for CAPM is a one rank update from scenario 1. It has to be noted that in this case the largest eigenvalue of $\bSigma$ is not bounded anymore. There is a little bit more noise, but there is no more temporal or structural dependence that we do not take care of in \eqref{eqn:obs}.

\begin{figure}
    \centering
    \includegraphics[width=\textwidth]{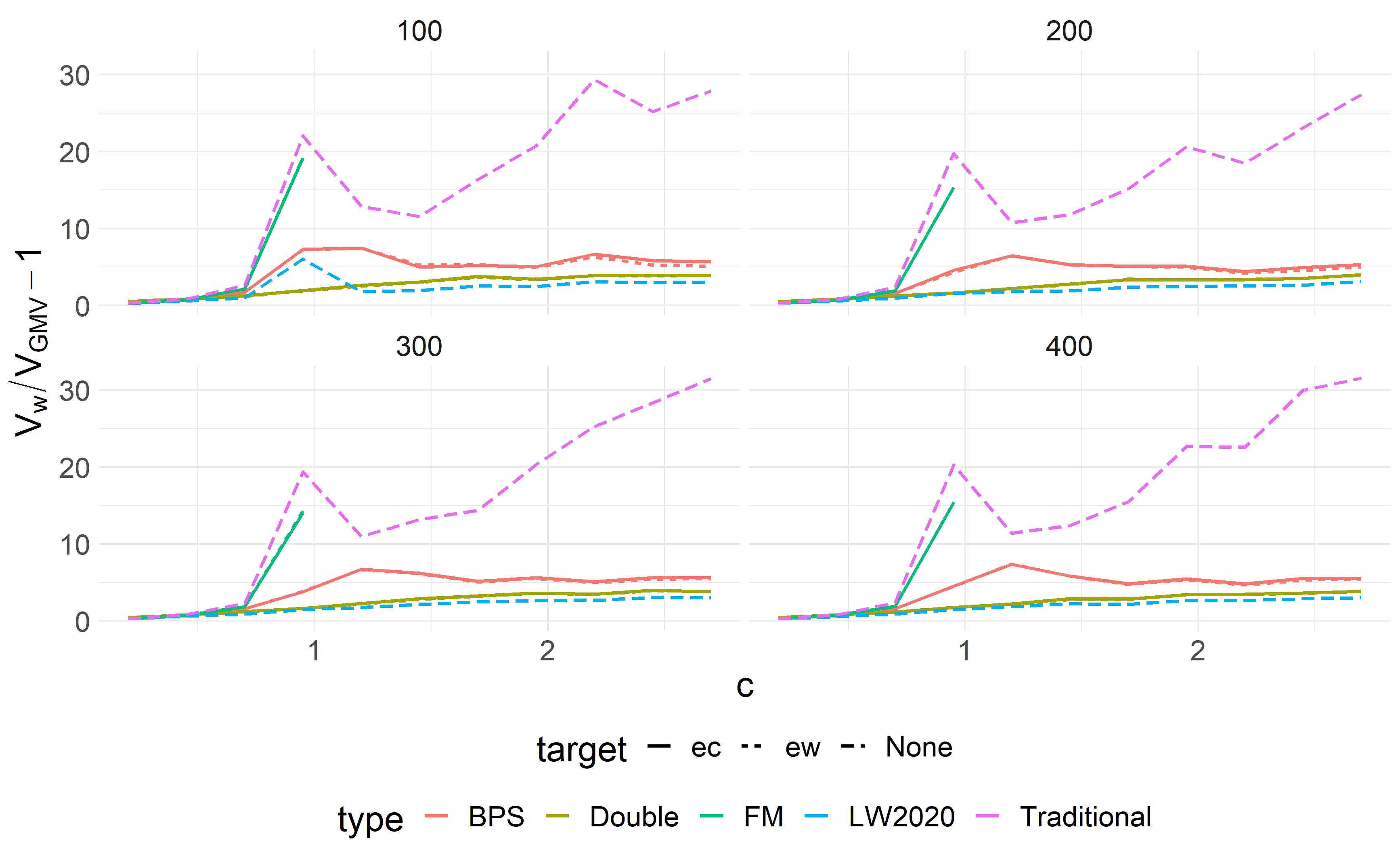}
    \caption{Relative loss $V_{\bw} / V_{GMV} - 1$ computed for several estimator of the GMV portfolio weights under scenario 2. Notice that some methods use different targets but are of the same type.}
    \label{fig:lossCAPMdist}
\end{figure}

Figure \ref{fig:lossCCCdist} depicts relative loses computed for the considered estimators of the GMV portfolio under scenario 3. It displays almost exactly the same plots as shown in Figures \ref{fig:lossTdist} and \ref{fig:lossCAPMdist}. The introduction of temporal dependence is not dramatic in terms of the relative loss. We still see the same type of ordering where Double and LW2020 are almost equally good when n is sufficiently large. 

\begin{figure}
    \centering
    \includegraphics[width=\textwidth]{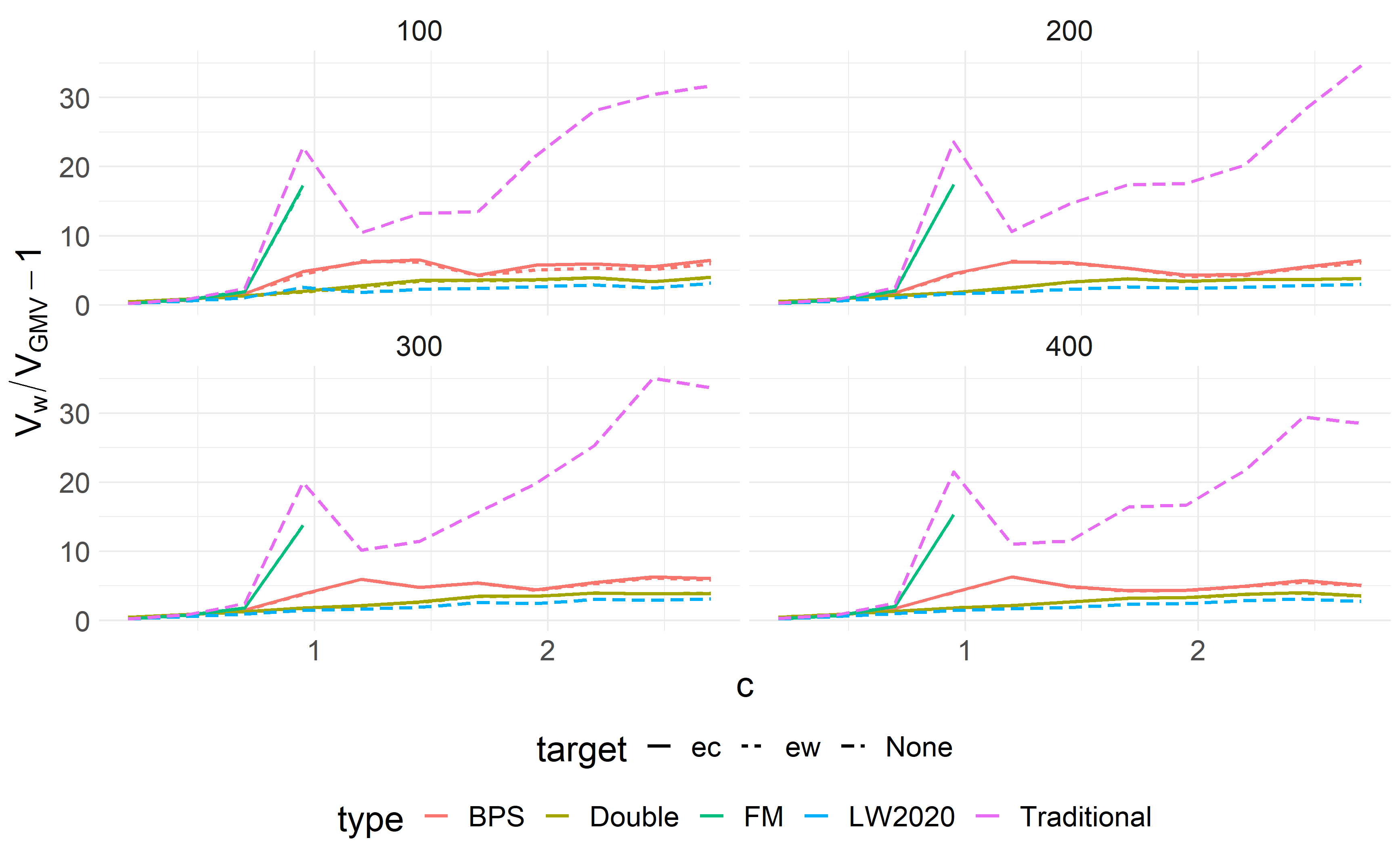}
    \caption{Relative loss $V_{\bw} / V_{GMV} - 1$ computed for several estimator of the GMV portfolio weights under scenario 3. Notice that some methods use different targets but are of the same type.}
    \label{fig:lossCCCdist}
\end{figure}

In Figure \ref{fig:lossVARMAdist} we can see the results obtained under scenario 4. Although the same ordering seems to hold the scale is not the same. Some methods, namely the Traditional, FM and BPS, have a much larger loss in comparison to previous scenarios. There is is a very little difference between the other methods for larger $c$ in contrast to previous scenarios. The difference between the two best methods, the Double and LW2020 portfolios, is as small as 1\% apart when $n=400$ and $c$ is large. 

\begin{figure}
    \centering
    \includegraphics[width=\textwidth]{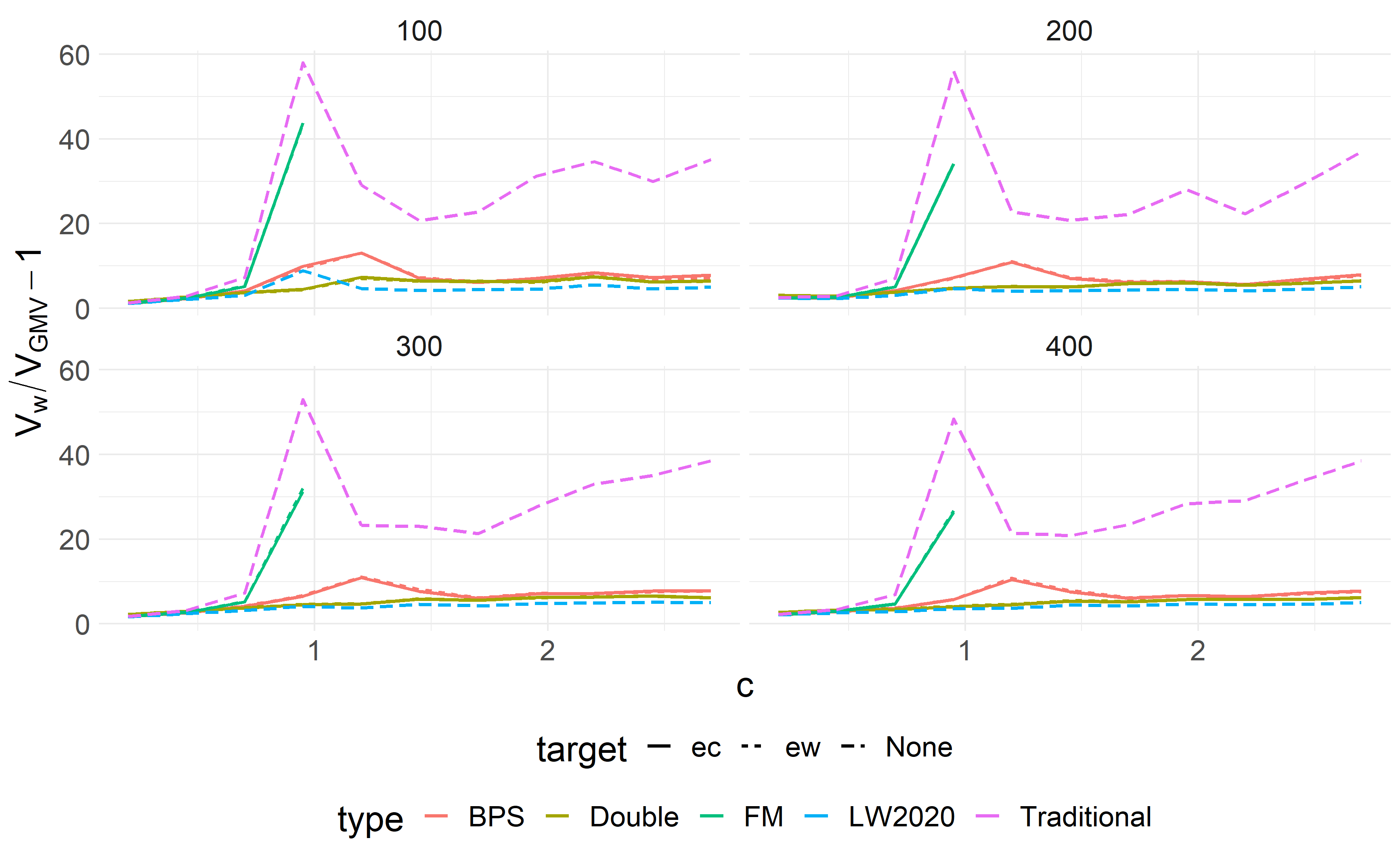}
    \caption{Relative loss $V_{\bw} / V_{GMV} - 1$ computed for several estimator of the GMV portfolio weights under scenario 4. Notice that some methods use different targets but are of the same type.}
    \label{fig:lossVARMAdist}
\end{figure}

All in all, the results of this simulation experiment justify that the proposed method is at least as good as the nonlinear shrinkage technique, which is already proved to be a state-of-the-art method for the estimation of large dimensional covariance matrices. Thus, it is of high importance to test it on a real data set using some other empirical measures of performance like the out-of-sample variance, return, Sharpe ratio, turnover etc.

\subsection{Empirical application}
In this section we will apply the different benchmark strategies on empirical data. The data constitute of daily (log returns) for $431$ assets from the S\&P500 index. The out-of-sample data ranges from 2013-01-01 to 2021-11-04. The in-sample data ranges back to early 2011. We will follow the previous section in that we use the equal correlation (ec) and equally weighted (ew) portfolios  as targets. In this empirical application we fix the window size to $n=250$ or $n=500$. We thereafter change the portfolio size $p$. The three different portfolio sizes we consider are $260$, $400$ and $431$. Thus, the window size $n=250$ reflects $c\in\{1.04, 1.6, 1.724\}$ and $n=500$ stands for $c\in\{0.52, 0.8, 0.862\}$.

All portfolios aim to minimize portfolio variance. Since the true portfolio variance is not available the most natural evaluation method should be which strategy provide the smallest out-of-sample variance, which will be denoted by $\sigma^{(k)}$. However, any portfolio is more than its volatility. A portfolio with small volatility does not necessarily provide a feasible return nor does it provide a feasible portfolio to invest in. We will therefore use the out-of-sample mean, which we denote $\bar{\by}^{(k)}_{\bw}$, as well as the out-of-sample Sharpe Ratio, denoted $SR^{(k)}$, to investigate the properties of the portfolio return distribution. Moreover, the stability of the portfolio weights also reflects how risky it is. To investigate the characteristics of the portfolio weights we will also consider the following performance measures
\begin{align}
    |\bw^{(k)}| &= \frac{1}{Tp} \sum_{i=1}^T\sum_{j=1}^p |w_{i,j}^{(k)}|,\label{portmes1}\\
    \max \bw^{(k)} &= \frac{1}{T} \sum_{i=1}^T \left( \max_j w_{i,j}^{(k)}\right),\label{portmes2}\\
    \min \bw^{(k)} &= \frac{1}{T} \sum_{i=1}^T \left( \min_j w_{i,j}^{(k)}\right),\label{portmes3}
\\
\bw_i^{(k)} \mathbbm{1}(\bw_i^{(k)} < 0) &= \frac{1}{T} \sum_{i=1}^T\sum_{j=1}^p w_{i,j}^{(k)}\mathbbm{1}(w_{i,j}^{(k)} < 0),\label{portmes4} \\
    \mathbbm{1}(\bw_i^{(k)} < 0) &= \frac{1}{Tp} \sum_{i=1}^T\sum_{j=1}^p \mathbbm{1}(w_{i,j}^{(k)} < 0)\label{portmes5}.
\end{align}

The first measure shown in equation \eqref{portmes1} is equal to the average size of the portfolio positions. A large value of this measure would indicate that the portfolio takes large positions (both negative and positive). It is a common critique to ordinary mean-variance portfolios because large positions are risky themselves. 
The second measure, shown in equation \eqref{portmes2}, is equal to the average long position. It is similar to the previous measure but only considers long positions of the portfolio. As with the above, small positions are preferred to large positions. 
Equation \eqref{portmes3} shows the average short position. With this measure we try to showcase how big short positions are. This is especially important since large short positions have potentially infinite risk. There is no limit to how much you can loose. Because of the great importance of how big short positions actually are we also include two further measures, which can be seen in  \eqref{portmes4} and \eqref{portmes5}. The former can be interpreted as the average size of the negative positions. The latter is the average proportion of short positions.

Note that for a portfolio with constant portfolio weights \eqref{portmes1}-\eqref{portmes5} are constant. For the equally weighted portfolio we have
\begin{equation}\label{eqn:ew_special_case}
    |\bw^{(k)}| = \max \bw^{(k)} = \min \bw^{(k)} = \frac{1}{p},\; \bw_i^{(k)} \mathbbm{1}(\bw_i^{(k)} < 0) = \mathbbm{1}(\bw_i^{(k)} < 0) = 0\,.
\end{equation}

In Table \ref{tab:charac_1} we display the results from the first experiment with window size equal to $250$ days. Due to \eqref{eqn:ew_special_case} we choose not to state the results for ew portfolio for \eqref{portmes1}-\eqref{portmes5}. The FM strategy has been removed since it has a very similar performance as the Traditional strategy for $c<1$ and it is not defined for $c>1$. For a moderately small portfolio $p=260$ the double shrinkage portfolio with equally correlated target provides the smallest out-of-sample variance which is denoted boldface. The double shrinkage with equally weighted target comes in second, indicated by the $^*$. There is a small difference between the different Double strategies using different targets, around 3 \%. However, they are all of similar performance except the double shrinkage with equally weighted target. It provides more stable weights with zero short positions and smallest turnover. The portfolio using LW2020 is ranked third in terms of volatility but is dominating other strategies in terms of mean and Sharpe ratio. The portfolio weights of nonlinear shrinkage, however, are not as stable as double shrinkage with ew target. 
The second best Sharpe ratio is provided by the Double with ew as a target. The difference with nonlinear shrinkage is around $0.002$. With the equally weighted target the Double is on par with LW2020. All others are far worse. The smallest Turnover is consistently provided by the Double shrinkage approach. The portfolio weights are very stable. Furthermore, the best performing portfolios, in terms of their characteristics, are  given by the Double shrinkage approach with ec or ew as a target portfolio. It is almost always a trade between them. These two take smaller positions on average, smaller short positions and less proportion of shorted weights. The natural ordering is that Traditional is worst and BPS being second to worst in terms of the portfolio characteristics. The LW2020 method comes in second, while the Double shrinkage portfolio being the best. 

When $p=400$ the equal correlation portfolio provides the best volatility estimate and equally weighted comes in second. LW2020 is the third best when it comes to volatility but is the best in terms of return. This brings it to being the best when it comes to the Sharpe ratio. Double shrinkage with ec comes in second. 
In this scenario we can also see that the Traditional portfolio is the most volatile but not as much as it could be expected. This is probably the Moore-Penrose inverse, which is still working reasonably well when $c$ is in the neighborhood of one. The Double shrinkage portfolio provides similar performance to the case $p=260$. It is always among the best performing. 

In the large dimensional case $p=431$, the most performing portfolio is the Double with equally correlated target when it comes to volatility. The estimation uncertainty has a large effect in these dimensions. However, second to best is the nonlinear shrinkage portfolio with slightly higher Sharpe ratio as double shrinkage. In terms of the mean, however, the best is the double shrinkage portfolio with equally weighted target. 
The Double with ew and ec as targets are always among the best and show very stable behaviour with a tiny turnover.

Next, we perform the same experiment with $n=500$. In Table \ref{tab:charac_2} we can see the results. When $p = 260$, we are in an intermediate concentration ratio case. It is equal to $c=260/500=0.52$. In this scenario the equally weighted portfolio is the best one in terms of volatility. The Double shrinkage approach with ew as a target is 5\% worse. As per usual, the Traditional estimator is the worst on estimating the volatility although not far behind in this small dimensional scenario. The portfolio with best average return is given by the Double shrinkage with ec as a target. Since it is also among the best in estimating volatility it gives the largest Sharpe ratio. The Double shrinkage with ew as a target is second to best. In terms of the portfolio characteristics the Double shrinkage approach with ew as a target dominates everything. It is the best in all scenarios. Similar picture is for $p=400$, here the best is the Double with ec as a target. The Traditional is showing surprisingly the largest return. However, it does so at a large cost (large variance and turnover).

\begingroup\fontsize{9}{11}\selectfont

\begin{longtable}[h!]{llllllll}
\caption{\label{tab:charac_1}Characteristics of the different strategies using a moving window approach. The out-of-sample period equals to 2234 days. The window size is held fixed, equal to 250.}\\
\toprule
\multicolumn{1}{c}{ } & \multicolumn{2}{c}{BPS} & \multicolumn{2}{c}{Double} & \multicolumn{3}{c}{ } \\
\cmidrule(l{3pt}r{3pt}){2-3} \cmidrule(l{3pt}r{3pt}){4-5}
name & ec & ew & ec & ew & ew & LW2020 & Traditional\\
\midrule
\addlinespace[0.3em]
\multicolumn{8}{l}{\textbf{p=260}}\\
\hspace{1em}$\sigma^k$ & $0.03912$ & $0.039$ & $\mathbf{0.01104}$ & $\mathit{0.01129}$* & $0.01133$ & $0.01165$ & $0.04597$\\
\hspace{1em}$\bar \by_{\bw}^k$ & $-0.000276$ & $-0.00025$ & $0.000397$ & $0.000503$ & $\mathit{0.000505}$* & $\mathbf{0.000542}$ & $-0.000352$\\
\hspace{1em}SR$^k$ & $-0.007$ & $-0.006$ & $0.036$ & $\mathit{0.045}$* & $\mathit{0.045}$* & $\mathbf{0.047}$ & $-0.008$\\
\hspace{1em}Turnover & $33267.6$ & $33703.26$ & $\mathit{69.74}$* & $\mathbf{3.36}$ &  & $2301.52$ & $41589.29$\\
\hspace{1em}$|\bw^{(k)}|$ & $0.0951$ & $0.0961$ & $\mathit{0.0066}$* & $\mathbf{0.0038}$ &  & $0.0102$ & $0.1189$\\
\hspace{1em}$\max \bw^{(k)}$ & $0.4494$ & $0.4496$ & $\mathit{0.0393}$* & $\mathbf{0.0046}$ &  & $0.0412$ & $0.5569$\\
\hspace{1em}$\min \bw^{(k)}$ & $-0.4209$ & $-0.4292$ & $\mathit{-0.0058}$* & $\mathbf{0.0028}$ &  & $-0.03$ & $-0.5327$\\
\hspace{1em}$\bw_i^{(k)} \mathbbm{1}(\bw_i^{(k)} < 0)$ & $-0.0929$ & $-0.0945$ & $\mathit{-0.0036}$* & $\mathbf{-2e-04}$ &  & $-0.0099$ & $-0.1171$\\
\hspace{1em}$\mathbbm{1}(\bw_i^{(k)} < 0)$ & $0.491$ & $0.488$ & $0.385$ & $\mathbf{0}$ &  & $\mathit{0.32}$* & $0.491$\\
\addlinespace[0.3em]
\multicolumn{8}{l}{\textbf{p=400}}\\
\hspace{1em}$\sigma^k$ & $0.01398$ & $0.01392$ & $\mathbf{0.01106}$ & $\mathit{0.01111}$* & $0.01112$ & $0.01193$ & $0.01524$\\
\hspace{1em}$\bar \by_{\bw}^k$ & $0.000525$ & $0.000464$ & $\mathit{0.00064}$* & $0.000514$ & $0.000509$ & $\mathbf{0.000706}$ & $0.000458$\\
\hspace{1em}SR$^k$ & $0.038$ & $0.033$ & $\mathit{0.058}$* & $0.046$ & $0.046$ & $\mathbf{0.059}$ & $0.03$\\
\hspace{1em}Turnover & $3661.98$ & $3678.92$ & $\mathit{69.71}$* & $\mathbf{4.81}$ &  & $667.62$ & $4776.77$\\
\hspace{1em}$|\bw^{(k)}|$ & $0.0192$ & $0.019$ & $\mathit{0.0043}$* & $\mathbf{0.0025}$ &  & $0.0107$ & $0.0246$\\
\hspace{1em}$\max \bw^{(k)}$ & $0.0788$ & $0.0769$ & $\mathit{0.0256}$* & $\mathbf{0.0032}$ &  & $0.0481$ & $0.0991$\\
\hspace{1em}$\min \bw^{(k)}$ & $-0.0699$ & $-0.07$ & $\mathit{-0.0041}$* & $\mathbf{0.0013}$ &  & $-0.0396$ & $-0.0913$\\
\hspace{1em}$\bw_i^{(k)} \mathbbm{1}(\bw_i^{(k)} < 0)$ & $-0.018$ & $-0.018$ & $\mathit{-0.0024}$* & $\mathbf{-6e-04}$ &  & $-0.0097$ & $-0.0235$\\
\hspace{1em}$\mathbbm{1}(\bw_i^{(k)} < 0)$ & $0.465$ & $0.46$ & $\mathit{0.388}$* & $\mathbf{0.004}$ &  & $0.423$ & $0.471$\\
\addlinespace[0.3em]
\multicolumn{8}{l}{\textbf{p=431}}\\
\hspace{1em}$\sigma^k$ & $0.0121$ & $0.01217$ & $\mathbf{0.01088}$ & $0.01117$ & $0.01116$ & $\mathit{0.01105}$* & $0.01322$\\
\hspace{1em}$\bar \by_{\bw}^k$ & $0.000437$ & $0.00042$ & $0.000496$ & $\mathbf{0.000516}$ & $0.000512$ & $\mathit{0.000515}$* & $0.000375$\\
\hspace{1em}SR$^k$ & $0.036$ & $0.034$ & $\mathit{0.046}$* & $\mathit{0.046}$* & $\mathit{0.046}$* & $\mathbf{0.047}$ & $0.028$\\
\hspace{1em}Turnover & $3165.55$ & $3192.65$ & $\mathit{69.99}$* & $\mathbf{5.19}$ &  & $660.56$ & $4254.45$\\
\hspace{1em}$|\bw^{(k)}|$ & $0.0164$ & $0.0161$ & $\mathit{0.004}$* & $\mathbf{0.0023}$ &  & $0.0101$ & $0.0214$\\
\hspace{1em}$\max \bw^{(k)}$ & $0.0663$ & $0.0645$ & $\mathit{0.0254}$* & $\mathbf{0.0031}$ &  & $0.0459$ & $0.0851$\\
\hspace{1em}$\min \bw^{(k)}$ & $-0.0589$ & $-0.0591$ & $\mathit{-0.0038}$* & $\mathbf{0.0011}$ &  & $-0.0381$ & $-0.0795$\\
\hspace{1em}$\bw_i^{(k)} \mathbbm{1}(\bw_i^{(k)} < 0)$ & $-0.0153$ & $-0.0152$ & $\mathit{-0.0022}$* & $\mathbf{-6e-04}$ &  & $-0.0093$ & $-0.0204$\\
\hspace{1em}$\mathbbm{1}(\bw_i^{(k)} < 0)$ & $0.46$ & $0.454$ & $\mathit{0.39}$* & $\mathbf{0.004}$ &  & $0.423$ & $0.468$\\
\bottomrule
\multicolumn{8}{l}{\rule{0pt}{1em}\textit{*} Second to best}\\
\end{longtable}
\endgroup{}

In the large dimensional case, where $p=431$, we can note that Double with ec as a target is now the best in terms of volatility whereas the ew portfolio is second to best. However, the difference is small. The equally weighted is good in general, as seen in previous examples but another explanation is that the window is to large for our data-driven portfolios to cope with changes. The BPS portfolio provides the largest mean with equal correlation target. The best Sharpe ratio is provided by the LW2020 with Double just slightly behind. 
Traditional still provides the largest return but in terms of variance, turnover and all other measures it is the most unstable and risky portfolio.
\begingroup\fontsize{9}{11}\selectfont

\begin{longtable}[h!]{llllllll}
\caption{\label{tab:charac_2}Characteristics of the different strategies using a moving window approach. The out-of-sample period equals to 2234 days. The window size is held fixed, equal to 500.}\\
\toprule
\multicolumn{1}{c}{ } & \multicolumn{2}{c}{BPS} & \multicolumn{2}{c}{Double} & \multicolumn{3}{c}{ } \\
\cmidrule(l{3pt}r{3pt}){2-3} \cmidrule(l{3pt}r{3pt}){4-5}
name & ec & ew & ec & ew & ew & LW2020 & Traditional\\
\midrule
\addlinespace[0.3em]
\multicolumn{8}{l}{\textbf{p=260}}\\
\hspace{1em}$\sigma^k$ & $0.01593$ & $0.01596$ & $0.01131$ & $\mathit{0.01123}$* & $\mathbf{0.01119}$ & $0.01215$ & $0.01821$\\
\hspace{1em}$\bar \by_{\bw}^k$ & $0.000439$ & $0.000372$ & $\mathbf{0.000621}$ & $0.000556$ & $0.000524$ & $\mathit{0.000562}$* & $0.00033$\\
\hspace{1em}SR$^k$ & $0.028$ & $0.023$ & $\mathbf{0.055}$ & $\mathit{0.05}$* & $0.047$ & $0.046$ & $0.018$\\
\hspace{1em}Turnover & $1016.3$ & $1039.9$ & $\mathit{224.46}$* & $\mathbf{221.3}$ &  & $442.06$ & $1385.55$\\
\hspace{1em}$|\bw^{(k)}|$ & $0.0259$ & $0.026$ & $\mathit{0.0117}$* & $\mathbf{0.0102}$ &  & $0.0172$ & $0.0344$\\
\hspace{1em}$\max \bw^{(k)}$ & $0.1798$ & $0.1742$ & $\mathit{0.0548}$* & $\mathbf{0.0429}$ &  & $0.0766$ & $0.2282$\\
\hspace{1em}$\min \bw^{(k)}$ & $-0.1213$ & $-0.1292$ & $\mathit{-0.0406}$* & $\mathbf{-0.0391}$ &  & $-0.0679$ & $-0.1738$\\
\hspace{1em}$\bw_i^{(k)} \mathbbm{1}(\bw_i^{(k)} < 0)$ & $-0.023$ & $-0.0243$ & $\mathit{-0.0093}$* & $\mathbf{-0.0086}$ &  & $-0.0154$ & $-0.0325$\\
\hspace{1em}$\mathbbm{1}(\bw_i^{(k)} < 0)$ & $0.479$ & $0.454$ & $\mathit{0.419}$* & $\mathbf{0.369}$ &  & $0.436$ & $0.47$\\
\addlinespace[0.3em]
\multicolumn{8}{l}{\textbf{p=400}}\\
\hspace{1em}$\sigma^k$ & $0.01547$ & $0.01517$ & $\mathbf{0.01117}$ & $0.0113$ & $\mathit{0.01119}$* & $0.01176$ & $0.02124$\\
\hspace{1em}$\bar \by_{\bw}^k$ & $\mathit{0.000779}$* & $0.000697$ & $0.000683$ & $0.000627$ & $0.000507$ & $0.000661$ & $\mathbf{0.000879}$\\
\hspace{1em}SR$^k$ & $0.05$ & $0.046$ & $\mathbf{0.061}$ & $0.055$ & $0.045$ & $\mathit{0.056}$* & $0.041$\\
\hspace{1em}Turnover & $2506.41$ & $2505.07$ & $\mathbf{283.58}$ & $\mathit{325.38}$* &  & $639.2$ & $4605.91$\\
\hspace{1em}$|\bw^{(k)}|$ & $0.024$ & $0.0235$ & $\mathbf{0.0082}$ & $\mathit{0.0084}$* &  & $0.0127$ & $0.0432$\\
\hspace{1em}$\max \bw^{(k)}$ & $0.146$ & $0.1387$ & $\mathit{0.0409}$* & $\mathbf{0.0383}$ &  & $0.0571$ & $0.2526$\\
\hspace{1em}$\min \bw^{(k)}$ & $-0.1129$ & $-0.118$ & $\mathbf{-0.0273}$ & $\mathit{-0.0317}$* &  & $-0.049$ & $-0.2176$\\
\hspace{1em}$\bw_i^{(k)} \mathbbm{1}(\bw_i^{(k)} < 0)$ & $-0.0219$ & $-0.0225$ & $\mathbf{-0.0069}$ & $\mathit{-0.0073}$* &  & $-0.0115$ & $-0.0418$\\
\hspace{1em}$\mathbbm{1}(\bw_i^{(k)} < 0)$ & $0.49$ & $0.465$ & $\mathit{0.415}$* & $\mathbf{0.403}$ &  & $0.443$ & $0.486$\\
\addlinespace[0.3em]
\multicolumn{8}{l}{\textbf{p=431}}\\
\hspace{1em}$\sigma^k$ & $0.01649$ & $0.01626$ & $\mathbf{0.01049}$ & $\mathit{0.01062}$* & $0.01116$ & $0.01099$ & $0.02591$\\
\hspace{1em}$\bar \by_{\bw}^k$ & $\mathit{0.000677}$* & $0.000627$ & $0.000534$ & $0.000537$ & $0.000512$ & $0.000573$ & $\mathbf{0.000957}$\\
\hspace{1em}SR$^k$ & $0.041$ & $0.039$ & $\mathit{0.051}$* & $\mathit{0.051}$* & $0.046$ & $\mathbf{0.052}$ & $0.037$\\
\hspace{1em}Turnover & $3202.26$ & $3217.62$ & $\mathbf{292.92}$ & $\mathit{327.96}$* &  & $764.9$ & $7195.79$\\
\hspace{1em}$|\bw^{(k)}|$ & $0.0228$ & $0.0222$ & $\mathbf{0.0077}$ & $\mathit{0.0078}$* &  & $0.0113$ & $0.0499$\\
\hspace{1em}$\max \bw^{(k)}$ & $0.1402$ & $0.1306$ & $\mathit{0.0393}$* & $\mathbf{0.036}$ &  & $0.0515$ & $0.2905$\\
\hspace{1em}$\min \bw^{(k)}$ & $-0.1076$ & $-0.1131$ & $\mathbf{-0.0263}$ & $\mathit{-0.0305}$* &  & $-0.0443$ & $-0.2588$\\
\hspace{1em}$\bw_i^{(k)} \mathbbm{1}(\bw_i^{(k)} < 0)$ & $-0.0206$ & $-0.0213$ & $\mathbf{-0.0065}$ & $\mathit{-0.0068}$* &  & $-0.0102$ & $-0.0483$\\
\hspace{1em}$\mathbbm{1}(\bw_i^{(k)} < 0)$ & $0.498$ & $0.466$ & $\mathit{0.416}$* & $\mathbf{0.401}$ &  & $0.439$ & $0.492$\\
\bottomrule
\multicolumn{8}{l}{\rule{0pt}{1em}\textit{*} Second to best}\\
\end{longtable}
\endgroup{}

\subsubsection{Tracking the S\&P500 index}
The last setting in Table \ref{tab:charac_1} and \ref{tab:charac_2}, where $p=431$ includes almost all of the stocks in the S\&P500 index. However, these portfolios are based on the assets in the S\&P 500 index today. Since we sample different assets from the index we may have a survival bias in the experiment above. That may have a positive effect on the return. A possibly more honest method is to choose the assets that were present in the index back in 2013. However, the index evolves and includes more assets today than it did before. This puts the application in another setting entirely. The portfolio size will need to change over time. We will therefore switch setting and track the S\&P500 index, trying to target the stocks that are available in the index. We choose the assets that are part of the index at the time and based on their availability (data quality). This excludes the survivorship bias in our result. Since we target the index we choose to reallocate whenever we register a change in the market capitalization. With $p$ changing over time, we choose two different window sizes. These are equal to $n=240, 720$. The market capitalization together with daily log returns we have at our disposal covers 406 assets in 2013 and 447 assets in late 2021. Our aim with this experiment is to see if we can improve the index volatility while still taking reasonable positions. We will therefore exclude all portfolios that do not use a target portfolio. We will only consider the BPS and Double estimators with the index-based target, and the index itself. The FM approach is excluded since $c>1$ for one of the scenarios.

\begin{wraptable}{r}{0.5\textwidth}
\caption{\label{tab:tab:charac_index} Out-of-sample results based on the moving window approach for the BPS and Double estimators with the index-based target, and for the index portfolio. The out-of-sample period equals to 1695 days. The portfolio size starts at 406 and is at most 447.}
\centering
\fontsize{9}{11}\selectfont
\begin{tabular}[t]{llll}
\toprule
name & BPS & Double & index \\
\midrule
\addlinespace[0.3em]
\multicolumn{4}{l}{\textbf{n=240}}\\
\hspace{1em}$\sigma^k$ & $\mathbf{0.00834}$ & $\mathit{0.00913}$* & $0.01066$\\
\hspace{1em}$\bar \by_{\bw}^k$ & $\mathbf{0.000622}$ & $\mathit{0.0006}$* & $0.000456$\\
\hspace{1em}SR$^k$ & $\mathbf{0.075}$ & $\mathit{0.066}$* & $0.043$\\
\hspace{1em}Turnover & $2555.37$ & $\mathit{185.93}$* & $\mathbf{7.04}$\\
\hspace{1em}$|\bw^{(k)}|$ & $0.0158$ & $\mathit{0.0041}$* & $\mathbf{0.0024}$\\
\hspace{1em}$\max \bw^{(k)}$ & $0.0642$ & $\mathbf{0.0387}$ & $\mathit{0.0506}$*\\
\hspace{1em}$\min \bw^{(k)}$ & $-0.0572$ & $\mathit{-0.0086}$* & $\mathbf{0}$\\
\hspace{1em}$\bw_i^{(k)} \mathbbm{1}(\bw_i^{(k)} < 0)$ & $-0.0148$ & $\mathit{-0.0056}*$ & $\mathbf{0}$ \\
\hspace{1em}$\mathbbm{1}(\bw_i^{(k)} < 0)$ & $0.453$ & $\mathit{0.158}$* & $\mathbf{0}$\\
\addlinespace[0.3em]
\multicolumn{4}{l}{\textbf{n=720}}\\
\hspace{1em}$\sigma^k$ & $\mathbf{0.00839}$ & $\mathit{0.00866}$* & $0.01066$\\
\hspace{1em}$\bar \by_{\bw}^k$ & $0.000231$ & $\mathit{0.000403}$* & $\mathbf{0.000456}$\\
\hspace{1em}SR$^k$ & $0.028$ & $\mathbf{0.047}$ & $\mathit{0.043}$*\\
\hspace{1em}Turnover & $1094.01$ & $\mathit{143.85}$* & $\mathbf{7.04}$\\
\hspace{1em}$|\bw^{(k)}|$ & $0.0204$ & $\mathit{0.0048}$* & $\mathbf{0.0024}$\\
\hspace{1em}$\max \bw^{(k)}$ & $0.1319$ & $\mathbf{0.039}$ & $\mathit{0.0506}$*\\
\hspace{1em}$\min \bw^{(k)}$ & $-0.106$ & $\mathit{-0.0131}$* & $\mathbf{0}$\\
\hspace{1em}$\bw_i^{(k)} \mathbbm{1}(\bw_i^{(k)} < 0)$ & $-0.0193$ & $\mathit{-0.0072}*$ & $\mathbf{0}$\\
\hspace{1em}$\mathbbm{1}(\bw_i^{(k)} < 0)$ & $0.465$ & $\mathit{0.165}$* & $\mathbf{0}$\\
\bottomrule
\end{tabular}
\end{wraptable}
In Table \ref{tab:tab:charac_index} we can see the results. When $n=240$, the portfolio will smallest volatility is BPS. It also provides the largest return. The double shrinkage approach is is the second to best while the index is the worst in both return and volatility. The same ordering holds for the Sharpe Ratio. Although the BPS provides the highest return, it does so at an extreme cost in comparison to the index and Double. If the investor is sensitive to large transitions then the Double method is a great middle-ground. It provides a very large decrease in the turnover as well as an decrease in the volatility, increase in mean and therefore an increase in the SR. The later weight characteristics displays the same behaviour as previously documented. The index has the smallest turnover. This can be explained by the fact that market cap is fairly stable and changes slowly over time. Thereafter its the Double being second to best with the exception of the largest long position. Double decreases volatility and seem to do so by taking smaller long positions, on average. It also introduces short positions, though relatively small ones. In comparison to the previous section, the Double with index as a target takes much smaller proportions of short positions.

When $n=720$, the same ordering holds for volatility as when $n=240$. The same does not hold for the mean. Now the index provides the best return with Double being 13\% worse. However, being the second to best in both mean and volatility makes the Double have highest SR. Increasing the the window size improves the stability in the portfolio weights and therefore a decrease in the turnover.





\section{Summary}\label{sec:sum}
In this paper we provide a novel method for investing in the GMV portfolio and a target portfolio. It uses a double shrinkage approach where the sample covariance matrix is shrunk with Thikonov regularization together with linear shrinkage of the GMV portfolio weights to a target portfolio. We construct a bona fide loss function which estimates the true loss function consistently. From that we estimate the two shrinkage coefficients given in the framework. The method is shown to be a great improvement over BPS and performs the same as LW2020 in an extensive simulation study. Furthermore, in the empirical application the method is shown to be a dominating investment strategy in majority of cases justified by different empirical performance measures.
We also show that it can act as a good portfolio to track an index. In this scenario it decreases the volatility but can still provide large Sharpe ratios. 
Our method is opinionated. That is, it demands the investors opinion on what a target portfolio is. That in turn implies that it will work best when the target portfolio is informative, in the sense of the investors aim. However, as our investigation shows, the investor can also use non-informative target portfolios and still achieve great results.

\section{Appendix}\label{sec:app}
For any integer $n>2$, we define
\begin{equation}\label{eqn:V_n}
    \bV_n=\frac{1}{n}\bX_n\left(\bI_{n}-\frac{1}{n}\ones_{n}\ones_{n}^\top\right)\bX_n^\top
\quad \text{and} \quad
\widetilde{\bV}_n=\frac{1}{n}\bX_n\bX_n^\top,
\end{equation}
where $\bX_n$ is given in \eqref{eqn:obs}. Hence,
\begin{equation}\label{eqn:one_rank_update}
    \bS_n 
        = \bSigma^{1/2}\bV_n\bSigma^{1/2}
        = \bSigma^{1/2} \widetilde{\bV}_n \bSigma^{1/2}
          -\bSigma^{1/2}\bar\bx_n\bar\bx_n^\top\bSigma^{1/2}
\end{equation}
with $\bar\bx_{n}=\frac{1}{n}\bX_n\ones_{n}$.

\vspace{1cm}
First, we present an important lemma which is a special case of Theorem 1 in \cite{rubmes2011}. 
Moreover, the following result (see, e.g., Theorem 1 on page 176 in \cite{ahlfors1953}) will be used in a sequel together with Lemma \ref{lem1_lam} in the proofs of the technical lemmas.

  \begin{thm}[Weierstrass]\label{weierstrass}
    Suppose that $f_n(z)$ is analytic in the region $\Omega_n$, and that the sequence $\{f_n(z)\}$ converges to a limit function $f(z)$ in a region $\Omega$, uniformly on every compact subset of $\Omega$. Then $f(z)$ is analytic in $\Omega$. Moreover, $f'(z)$ converges uniformly to $f'(z)$ on every compact subset of $\Omega$.
  \end{thm}


We will need two interchange the limits and derivatives many times that is why Theorem \ref{weierstrass} plays a vital role here.
More on the application of Weierstrass theorem can be found in the appendix of \cite{bop2021}.

\begin{lem}\label{lem1_lam}
Let a nonrandom $p\times p$-dimensional matrix $\mathbf{\Theta}_p$ possess a uniformly bounded trace norm. Then it holds that
\begin{enumerate}[(i)]
\item
\begin{equation}\label{RM2011_id_lam1}
\left|
        \text{tr}\left(\mathbf{\Theta}_p\left(\dfrac{1}{n}\bX_n\bX_n^{\prime}+\eta \bSigma^{-1}-z\bI_p\right)^{-1}\right)
        -\text{tr}\left(\mathbf{\Theta}_p (\eta\bSigma^{-1}+(v(\eta,z)-z)\bI)^{-1}\right)
    \right|\stackrel{a.s.}{\longrightarrow}0
\end{equation}
for $p/n\longrightarrow c \in (0, +\infty)$ as $n\rightarrow\infty$ where $v(z)$ solves the following equality
\begin{equation}\label{v1z_lam}
    v(\eta, z)=\frac{1}{1+c\frac{1}{p}\text{tr}\left((\eta\bSigma^{-1}+(v(\eta, z)-z)\bI)^{-1}\right)}.
\end{equation}
\item
\begin{eqnarray}\label{RM2011_id_lam_2a}
&&
\Bigg|
        \text{tr}\left(\mathbf{\Theta}_p\left(\dfrac{1}{n}\bX_n\bX_n^{\prime}+\eta\bSigma^{-1}-z\bI_p\right)^{-1}\bSigma^{-1}
        \left(\dfrac{1}{n}\bX_n\bX_n^{\prime}+\eta\bSigma^{-1}-z\bI_p\right)^{-1}\right) \nonumber\\
&-&\text{tr}\left(\mathbf{\Theta}_p (\eta\bSigma^{-1}+(v(\eta, z)-z)\bI)^{-1} \bSigma^{-1} (\eta\bSigma^{-1}+(v(\eta, z)-z)\bI)^{-1}\right) \nonumber\\
&-&v^\prime_1(\eta, z)\text{tr}\left(\mathbf{\Theta}_p (\eta\bSigma^{-1}+(v(\eta, z)-z)\bI)^{-2}\right)
    \Bigg|\stackrel{a.s.}{\longrightarrow}0
\end{eqnarray}
for $p/n\longrightarrow c \in (0, +\infty)$ as $n\rightarrow\infty$ with
\begin{equation}\label{v2az_lam}
    v^\prime_1(\eta, z)=
        \frac{-
            \frac{1}{p}\text{tr}\left((\eta\bSigma^{-1}+(v(\eta, z)-z)\bI)^{-1} \bSigma^{-1} (\eta\bSigma^{-1}+(v(\eta, z)-z)\bI)^{-1}\right)
        }{
            \frac{1}{p}\text{tr}\left((\eta\bSigma^{-1}+(v(\eta, z)-z)\bI)^{-2}\right)-c^{-1}v(\eta, z)^{-2}
        }.
\end{equation}
\item
\begin{eqnarray}\label{RM2011_id_lam2}
&&
\Bigg|
        \text{tr}\left(\mathbf{\Theta}_p\left(\dfrac{1}{n}\bX_n\bX_n^{\prime}+\eta\bSigma^{-1}-z\bI_p\right)^{-2}\right)  \nonumber \\
&-&(1-v^\prime_2(\eta, z))\text{tr}\left(\mathbf{\Theta}_p (\eta\bSigma^{-1}+(v(\eta, z)-z)\bI)^{-2}\right)
    \Bigg|\stackrel{a.s.}{\longrightarrow}0
\end{eqnarray}
for $p/n\longrightarrow c \in (0, +\infty)$ as $n\rightarrow\infty$ with
\begin{equation}\label{v2z_lam}
    v^\prime_2(\eta, z)=
        \frac{
            \frac{1}{p}\text{tr}\left((\eta\bSigma^{-1}+(v(\eta, z)-z)\bI)^{-2}\right)
        }{
            \frac{1}{p}\text{tr}\left((\eta\bSigma^{-1}+(v(\eta, z)-z)\bI)^{-2}\right)-c^{-1}v(\eta, z)^{-2}
        }.
\end{equation}
\end{enumerate}
\end{lem}

\begin{proof}[Proof of Lemma \ref{lem1_lam}:]
\begin{enumerate}[(i)]
    \item The application of Theorem 1 in \cite{rubmes2011} leads to (\ref{RM2011_id_lam1}) where $v(\eta, z)$ is a unique solution in $\mathbbm{C}^+$ of the following equation
\begin{equation}\label{eq1-Lemma6_1_lam}
    \dfrac{1}{v(\eta, z)}-1
    =\frac{c}{p}\text{tr}\left((\eta\bSigma^{-1}+(v(\eta, z)-z)\bI)^{-1}\right)\,.
\end{equation}

\item For the second result of the lemma we get that
\begin{eqnarray*}
&&\text{tr}\left(\mathbf{\Theta}_p\left(\dfrac{1}{n}\bX_n\bX_n^{\prime}+\eta\bSigma^{-1}-z\bI_p\right)^{-1}\bSigma^{-1}
        \left(\dfrac{1}{n}\bX_n\bX_n^{\prime}+\eta\bSigma^{-1}-z\bI_p\right)^{-1}\right)\\
&=&-\dfrac{\partial}{\partial \eta}\text{tr}\left(\mathbf{\Theta}_p \left(\dfrac{1}{n}\bX_n\bX_n^{\prime}+\eta\bSigma^{-1}-z\bI_p\right)^{-1}\right),
\end{eqnarray*}
which almost surely converges to
\begin{eqnarray*}
&&-\dfrac{\partial}{\partial \eta}\text{tr}\left(\mathbf{\Theta}_p (\eta\bSigma^{-1}+(v(\eta, z)-z)\bI)^{-1}\right)\\
&=& \text{tr}\left(\mathbf{\Theta}_p (\eta\bSigma^{-1}+(v(\eta, z)-z)\bI)^{-1}
(\bSigma^{-1}+v_1^\prime(\eta, z)\bI)
(\eta\bSigma^{-1}+(v(\eta, z)-z)\bI)^{-1}\right)
\end{eqnarray*}
following Theorem \ref{weierstrass}. The first-order partial derivative $v_1^\prime(\eta, z)$ is obtained from \eqref{eq1-Lemma6_1_lam} as
\begin{eqnarray*}
&&-\dfrac{v_1^\prime(\eta, z)}{v(\eta, z)^2}
=-\frac{c}{p}\text{tr}\left((\eta\bSigma^{-1}+(v(\eta, z)-z)\bI)^{-1}
    \bSigma^{-1}
    (\eta\bSigma^{-1}+(v(\eta, z)-z)\bI)^{-1}\right)\\
&-&v_1^\prime(\eta, z)\frac{c}{p}\text{tr}\left((\eta\bSigma^{-1}+(v(\eta, z)-z)\bI)^{-2}\right),
\end{eqnarray*}
from which \eqref{v2az_lam} is deduced.

\item For the third assertion of the lemma we note that
\[\text{tr}\left(\mathbf{\Theta}_p\left(\dfrac{1}{n}\bX_n\bX_n^{\prime}+\eta\bSigma^{-1}-z\bI_p\right)^{-2}\right)=\dfrac{\partial}{\partial z}\text{tr}\left(\mathbf{\Theta}_p \left(\dfrac{1}{n}\bX_n\bX_n^{\prime}+\eta\bSigma^{-1}-z\bI_p\right)^{-1}\right),
\]
which almost surely tends to
\[\dfrac{\partial}{\partial z}\text{tr}\left(\mathbf{\Theta}_p (\eta\bSigma^{-1}+(v(\eta, z)-z)\bI)^{-1}\right)=
(1-v_2 ^\prime(\eta, z)) \text{tr}\left(\mathbf{\Theta}_p (\eta\bSigma^{-1}+(v(\eta, z)-z)\bI)^{-2}\right)
\]
following Theorem \ref{weierstrass}. Moreover, $v_2^\prime(\eta, z)$ is computed from \eqref{eq1-Lemma6_1_lam} and it is obtained from the following equation
\begin{equation*}
    -\dfrac{v_2^\prime(\eta, z)}{v(\eta, z)^2}
    =(1-v_2^\prime(\eta, z))\frac{c}{p}\text{tr}\left((\eta\bSigma^{-1}+(v(\eta, z)-z)\bI)^{-2}\right).
\end{equation*}

This completes the proof of the lemma.
\end{enumerate}
\end{proof}

\begin{lem}\label{lem2_lam}
Let $\boldsymbol{\theta}$ and $\boldsymbol{\xi}$ be universal nonrandom vectors with bounded Euclidean norms. Then it holds that
\begin{eqnarray}
  &&  \left|
        \boldsymbol{\xi}^\prime\left(\dfrac{1}{n}\bX_n\bX_n^{\prime}+\eta \bSigma^{-1}\right)^{-1}\boldsymbol{\theta}
        - \boldsymbol{\xi}^\prime (\eta\bSigma^{-1}+v(\eta, 0)\bI)^{-1} \boldsymbol{\theta}
    \right| \stackrel{a.s.}{\longrightarrow} 0 \,,\label{1_lam}\\
&&\Bigg|
        \boldsymbol{\xi}^\prime\left(\dfrac{1}{n}\bX_n\bX_n^{\prime}+\eta \bSigma^{-1}\right)^{-1}\bSigma^{-1}
        \left(\dfrac{1}{n}\bX_n\bX_n^{\prime}+\eta \bSigma^{-1}\right)^{-1}\boldsymbol{\theta}\label{2a_lam}\\
&& -\boldsymbol{\xi}^\prime (\eta\bSigma^{-1}+v(\eta, 0)\bI)^{-1} 
\bSigma^{-1}(\eta\bSigma^{-1}+v(\eta, 0)\bI)^{-1}\boldsymbol{\theta}
-v_1^\prime(\eta, 0)\boldsymbol{\xi}^\prime (\eta\bSigma^{-1}+v(\eta, 0)\bI)^{-2}\boldsymbol{\theta}
    \Bigg| \stackrel{a.s.}{\longrightarrow} 0 \nonumber \\
&&\left|
        \boldsymbol{\xi}^\prime\left(\dfrac{1}{n}\bX_n\bX_n^{\prime}+\eta \bSigma^{-1}\right)^{-2}\boldsymbol{\theta}-
        (1-v_2^\prime(\eta, 0))\boldsymbol{\xi}^\prime (\eta\bSigma^{-1}+v(\eta, 0)\bI)^{-2}\boldsymbol{\theta}
    \right| \stackrel{a.s.}{\longrightarrow} 0 \label{2_lam}
\end{eqnarray}
for $p/n\longrightarrow c \in (0,\infty)$ as $n\rightarrow\infty$ where
$v(\eta, 0)$ is the solution of
\begin{equation}\label{v_eta_lam}
v(\eta, 0) = 1-c\left(1-\frac{\eta}{p}\tr\left(\left(v(\eta,0)\bSigma+\eta\bI \right)^{-1}\right) \right),
\end{equation}
and $v_1^\prime(\eta, 0)$ and $v_2^\prime(\eta, 0)$ are computed by
\begin{equation}\label{v1_eta_lam}
v_1^\prime(\eta, 0)  = v(\eta,0)\frac{c\frac{1}{p}\text{tr}\left((v(\eta,0)\bSigma+\eta\bI)^{-1}\right)-c\eta \frac{1}{p}\text{tr}\left((v(\eta,0)\bSigma+\eta\bI)^{-2}\right)}
{1-c+2c\eta\frac{1}{p}\text{tr}\left((v(\eta,0)\bSigma+\eta\bI)^{-1}\right)-
c\eta^2\frac{1}{p}\text{tr}\left((v(\eta,0)\bSigma+\eta\bI)^{-2}\right)}
\end{equation}
and
\begin{equation}\label{v2_eta_lam}
v_2^\prime(\eta, 0)  =1-\frac{1}{v(\eta,0)}+\eta\frac{v_1^\prime(\eta,0)}{v(\eta,0)^2}.
\end{equation}
\end{lem}

\begin{proof}[Proof of Lemma \ref{lem2_lam}:]
Since the trace norm of $\btheta\bxi^\prime$ is uniformly bounded, i.e.,
\[||\btheta\bxi^\prime||_{tr}\le \sqrt{\btheta^\prime\btheta}\sqrt{\bxi^\prime\bxi} <\infty,\]
the application of Lemma \ref{lem1_lam} leads to \eqref{1_lam}, \eqref{2a_lam}, and \eqref{2_lam} where $v(\eta,0)$ satisfies the following equality
  \begin{equation*}
    \frac{1}{v(\eta, 0)} -1= 
        \frac{c}{p}\tr \left( \left(\eta\bSigma^{-1} + v\left(\eta, 0\right)\bI \right)^{-1} \right) = \frac{c}{v(\eta, 0)}\left(1-\frac{\eta}{p}\tr\left(\left(v(\eta,0)\bSigma+\eta\bI \right)^{-1}\right) \right),
    \end{equation*}
which results in \eqref{v_eta_lam}.
    
The application of \eqref{v2az_lam} leads to
\begin{eqnarray*}
v^\prime_1(\eta, 0)&=& \frac{-
            \frac{1}{p}\text{tr}\left((\eta\bSigma^{-1}+v(\eta, 0)\bI)^{-1} \bSigma^{-1} (\eta\bSigma^{-1}+v(\eta, 0)\bI)^{-1}\right)} {\frac{1}{p}\text{tr}\left((\eta\bSigma^{-1}
            +v(\eta,0)\bI)^{-2}\right)-c^{-1}v(\eta,0)^{-2}}\\
        &=&v(\eta,0) \frac{c\frac{1}{p}\text{tr}\left(v(\eta,0)(\eta\bSigma^{-1}+v(\eta, 0)\bI)^{-1} \bSigma^{-1} (\eta\bSigma^{-1}+v(\eta, 0)\bI)^{-1}\right)}
        {1-c\frac{1}{p}\text{tr}\left(v(\eta, 0)^2(\eta\bSigma^{-1}+v(\eta, 0)\bI)^{-2}\right)}\\
&=&v(\eta,0)\frac{c\frac{1}{p}\text{tr}\left((v(\eta,0)\bSigma+\eta\bI)^{-1}\right)-c\eta \frac{1}{p}\text{tr}\left((v(\eta,0)\bSigma+\eta\bI)^{-2}\right)}
{1-c+2c\eta\frac{1}{p}\text{tr}\left((v(\eta,0)\bSigma+\eta\bI)^{-1}\right)-
c\eta^2\frac{1}{p}\text{tr}\left((v(\eta,0)\bSigma+\eta\bI)^{-2}\right)}.
\end{eqnarray*}

Finally, using \eqref{v2z_lam}, we get
\begin{eqnarray*}
    v^\prime_2(\eta, 0)&=&
        \frac{
            \frac{1}{p}\text{tr}\left((\eta\bSigma^{-1}+v(\eta,0)\bI)^{-2}\right)
        }{\frac{1}{p}\text{tr}\left((\eta\bSigma^{-1}+v(\eta, 0)\bI)^{-2}\right)-c^{-1}v(\eta, 0)^{-2}
        }\\
&=&1-\frac{1}{1-c\frac{1}{p}\text{tr}\left(v(\eta, 0)^2(\eta\bSigma^{-1}+v(\eta, 0)\bI)^{-2}\right)}\\
&=&1-\frac{1}{1-c
+2c\eta\frac{1}{p}\text{tr}\left((v(\eta,0)\bSigma+\eta\bI)^{-1}\right)-
c\eta^2\frac{1}{p}\text{tr}\left((v(\eta,0)\bSigma+\eta\bI)^{-2}\right)}\\
&=&1-\frac{1}{v(\eta,0)+
\eta\left(c\frac{1}{p}\text{tr}\left((v(\eta,0)\bSigma+\eta\bI)^{-1}\right)-
c\eta\frac{1}{p}\text{tr}\left((v(\eta,0)\bSigma+\eta\bI)^{-2}\right)\right)}
\\
&=&1-\frac{1}{v(\eta,0)+\eta\frac{v(\eta,0) v_1^\prime(\eta,0)}{v(\eta,0)-\eta v_1^\prime(\eta,0)}}=1-\frac{1}{v(\eta,0)}+\eta\frac{v_1^\prime(\eta,0)}{v(\eta,0)^2}.
\end{eqnarray*}
\end{proof}

\begin{lem}\label{lem:lem3}
Let $\boldsymbol{\theta}$ and $\boldsymbol{\xi}$ be universal nonrandom vectors such that $\bSigma^{-1/2}\boldsymbol{\theta}$ and $\bSigma^{-1/2}\boldsymbol{\xi}$ have bounded Euclidean norms. Then it holds that
\begin{align}
        \left|
            \bxi^\top \bS_{\lambda}^{-1}\btheta  - 
            \bxi^\top\bOmega_{\lambda}^{-1}\btheta
        \right| &\stackrel{a.s.}{\longrightarrow} 0, \label{lem3_eq1}\\
         \left|
            \bxi^\top  \bS_{\lambda}^{-2} \btheta 
            -  \bxi^\top  \bOmega_{\lambda}^{-2} \btheta 
            -v_1^\prime(\eta, 0) \bxi^\top  \bOmega_{\lambda}^{-1} \bSigma \bOmega_{\lambda}^{-1} \btheta 
        \right| &\stackrel{a.s.}{\longrightarrow} 0,\label{lem3_eq2}\\
        \left|
            \bxi^\top  \bS_{\lambda}^{-1} \bSigma \bS_{\lambda}^{-1} \btheta 
            - (1-v_2^\prime(\eta, 0))
            \bxi^\top  \bOmega_{\lambda}^{-1} \bSigma \bOmega_{\lambda}^{-1} \btheta 
        \right| &\stackrel{a.s.}{\longrightarrow} 0 \label{lem3_eq3}
    \end{align}
    for $p/n\longrightarrow c \in (0,\infty)$ as $n\rightarrow\infty$ with $\eta=1/\lambda-1$,
    \[\bOmega_{\lambda}= v\left(\eta,0\right)\lambda\bSigma + (1-\lambda)\bI,\]
and $v(\eta,0)$, $v_1^\prime(\eta, 0)$ and $v_2^\prime(\eta, 0)$ given in Lemma \ref{lem2_lam}.
\end{lem}
\begin{proof}[Proof of Lemma \ref{lem:lem3}:]
Let $\widetilde{\bS}_n=\bSigma^{1/2}\widetilde{\bV}_n\bSigma^{1/2}$. Using \eqref{eqn:V_n} and \eqref{eqn:one_rank_update} and the formula for the 1-rank update of inverse matrix (see, e.g., \cite{hornjohn1985}), we get 
    \begin{align*}
        &\lambda\bxi^\top \bS_{\lambda}^{-1} \btheta =  \bxi^\top \left(\widetilde{\bS}_n + \left(\frac{1}{\lambda} - 1\right)\bI
          -\bSigma^{1/2}\bar\bx_n\bar\bx_n^\top\bSigma^{1/2} \right)^{-1} \btheta \\
          &= \bxi^\top \bSigma^{-1/2}\left( \widetilde{\bV}_n  + \left(\frac{1}{\lambda} - 1\right)\bSigma^{-1}\right)^{-1} \bSigma^{-1/2}\btheta
          \\
          & +\frac{\bxi^\top\left(\widetilde{\bS}_n + \left(\frac{1}{\lambda} - 1\right)\bI\right)^{-1}\bSigma^{1/2}\bar\bx_n
                \bar\bx_n^\top\bSigma^{1/2} \left(\widetilde{\bS}_n + \left(\frac{1}{\lambda} - 1\right)\bI\right)^{-1}\btheta
            }{1 - \bar\bx_n^\top\bSigma^{1/2}  \left(\widetilde{\bS}_n + \left(\frac{1}{\lambda} - 1\right)\bI\right)^{-1} \bSigma^{1/2}\bar\bx_n
            },
    \end{align*}
where
\begin{equation}\label{eqn:mean_stieltjes_conv_1}
\left |\bxi^\top\left(\widetilde{\bS}_n + \left(\frac{1}{\lambda} - 1\right)\bI\right)^{-1}\bSigma^{1/2}\bar\bx_n
\right|\stackrel{a.s.}{\longrightarrow} 0
\end{equation}
for $\lambda \in (0,1]$ by \citet[p. 673]{pan2014comparison}.
Furthermore, the quantity
\begin{align}\label{eqn:mean_stieltjes_conv_2}
    \frac{1
        }{1 - \bar\bx_n^\top\bSigma^{1/2}  \left(\widetilde{\bS}_n + \left(\frac{1}{\lambda} - 1\right)\bI\right)^{-1} \bSigma^{1/2}\bar\bx_n
        }
    \end{align}
is bounded following \citet[Eq. (2.28)]{pan2014comparison}.
Hence, the application of Lemma \ref{lem2_lam} leads to the first statement of Lemma \ref{lem:lem3}.

We compute
{\footnotesize
\begin{align*}
        &\lambda^2 \bxi^\top \bS_{\lambda}^{-2} \btheta =
        \bxi^\top \left(\widetilde{\bS}_n +
        \left(\frac{1}{\lambda} - 1\right)\bI
          -\bar\bSigma^{1/2}\bx_n\bar\bx_n^\top\bSigma^{1/2} \right)^{-2}\btheta\\
        &=
        \bxi^\top \bSigma^{-1/2} \left(\widetilde{\bV}_n  + \left(\frac{1}{\lambda} - 1\right)\bSigma^{-1}
         \right)^{-1} \bSigma^{-1}
         \left(\widetilde{\bV}_n  + \left(\frac{1}{\lambda} - 1\right)\bSigma^{-1}
         \right)^{-1}
         \bSigma^{-1/2}\btheta\\
        &+\frac{\bxi^\top \left(\widetilde{\bS}_n  + \left(\frac{1}{\lambda} - 1\right)\bI
         \right)^{-2}\bSigma^{1/2}\bar\bx_n
                \bar\bx_n^\top \bSigma^{1/2} \left(\widetilde{\bS}_n
                + \left(\frac{1}{\lambda} - 1\right)\bI
         \right)^{-1}\btheta
            }{1 - \bar\bx_n^\top\bSigma^{1/2}  \left(\widetilde{\bS}_n + \left(\frac{1}{\lambda} - 1\right)\bI\right)^{-1} \bSigma^{1/2}\bar\bx_n
            }\\
        &+\frac{\bxi^\top \left(\widetilde{\bS}_n  + \left(\frac{1}{\lambda} - 1\right)\bI
         \right)^{-1}\bSigma^{1/2}\bar\bx_n
                \bar\bx_n^\top \bSigma^{1/2} \left(\widetilde{\bS}_n
                + \left(\frac{1}{\lambda} - 1\right)\bI
         \right)^{-2}\btheta
            }{1 - \bar\bx_n^\top\bSigma^{1/2}  \left(\widetilde{\bS}_n + \left(\frac{1}{\lambda} - 1\right)\bI\right)^{-1} \bSigma^{1/2}\bar\bx_n
            }\\
        &+
         \bar\bx_n^\top\bSigma^{1/2}\left(\widetilde{\bS}_n  + \left(\frac{1}{\lambda} - 1\right)\bI
         \right)^{-2}\bSigma^{1/2}\bar\bx_n\\
         &\times
         \frac{\bxi^\top \left(\widetilde{\bS}_n  + \left(\frac{1}{\lambda} - 1\right)\bI
         \right)^{-1}\bSigma^{1/2}\bar\bx_n
                \bar\bx_n^\top \bSigma^{1/2} \left(\widetilde{\bS}_n
                + \left(\frac{1}{\lambda} - 1\right)\bI
         \right)^{-1}\btheta
            }{
               \left(1 - \bar\bx_n^\top\bSigma^{1/2}  \left(\widetilde{\bS}_n + \left(\frac{1}{\lambda} - 1\right)\bI\right)^{-1} \bSigma^{1/2}\bar\bx_n\right)^2
            },
\end{align*}
}
where
{\footnotesize
\begin{eqnarray*}
&&\bar\bx_n^\top\bSigma^{1/2}\left(\widetilde{\bS}_n  + \left(\frac{1}{\lambda} - 1\right)\bI        \right)^{-2}\bSigma^{1/2}\bar\bx_n\\
&\le& \left(\frac{1}{\lambda} - 1\right)^{-1} \bar\bx_n^\top\bSigma^{1/2}\left(\widetilde{\bS}_n  + \left(\frac{1}{\lambda} - 1\right)\bI        \right)^{-1}\bSigma^{1/2}\bar\bx_n <\infty 
\end{eqnarray*}
}
and
{\footnotesize
\begin{eqnarray*}
&&\bxi^\top\left(\widetilde{\bS}_n  + \left(\frac{1}{\lambda} - 1\right)\bI\right)^{-2}\bSigma^{1/2}\bar\bx_n\\
&\le&\sqrt{\bxi^\top\left(\widetilde{\bS}_n  + \left(\frac{1}{\lambda} - 1\right)\bI\right)^{-2}\bxi}
\sqrt{\bar\bx_n^\top\bSigma^{1/2}\left(\widetilde{\bS}_n  + \left(\frac{1}{\lambda} - 1\right)\bI        \right)^{-2}\bSigma^{1/2}\bar\bx_n}
 <\infty 
\end{eqnarray*}
}
For the third statement of the lemma we consider
{\footnotesize
\begin{align*}
        &\lambda^2 \bxi^\top \bS_{\lambda}^{-1} \bSigma \bS_{\lambda}^{-1} \btheta =
        \bxi^\top \bSigma^{-1/2} \left(\widetilde{\bV}_n  + \left(\frac{1}{\lambda} - 1\right)\bSigma^{-1}
          -\bar\bx_n\bar\bx_n^\top \right)^{-2} \bSigma^{-1/2}\btheta\\
        &=
        \bxi^\top \bSigma^{-1/2} \left(\widetilde{\bV}_n  + \left(\frac{1}{\lambda} - 1\right)\bSigma^{-1}
         \right)^{-2} \bSigma^{-1/2}\btheta\\
        &+\frac{\bxi^\top\bSigma^{-1/2} \left(\widetilde{\bV}_n  + \left(\frac{1}{\lambda} - 1\right)\bSigma^{-1}
         \right)^{-2}\bar\bx_n
                \bar\bx_n^\top  \left(\widetilde{\bV}_n  + \left(\frac{1}{\lambda} - 1\right)\bSigma^{-1}
         \right)^{-1}\bSigma^{-1/2}\btheta
            }{1 - \bar\bx_n^\top\bSigma^{1/2}  \left(\widetilde{\bS}_n+ \left(\frac{1}{\lambda} - 1\right)\bI\right)^{-1} \bSigma^{1/2}\bar\bx_n
            }\\
        &+\frac{\bxi^\top\bSigma^{-1/2} \left(\widetilde{\bV}_n  + \left(\frac{1}{\lambda} - 1\right)\bSigma^{-1}
         \right)^{-1}\bar\bx_n
                \bar\bx_n^\top  \left(\widetilde{\bV}_n  + \left(\frac{1}{\lambda} - 1\right)\bSigma^{-1}
         \right)^{-2}\bSigma^{-1/2}\btheta
            }{1 - \bar\bx_n^\top\bSigma^{1/2}  \left(\widetilde{\bS}_n + \left(\frac{1}{\lambda} - 1\right)\bI\right)^{-1} \bSigma^{1/2}\bar\bx_n
            }\\
        &+
         \bar\bx_n^\top\left(\widetilde{\bV}_n  + \left(\frac{1}{\lambda} - 1\right)\bSigma^{-1}
         \right)^{-2}\bar\bx_n\\
         &\times
         \frac{\bxi^\top\bSigma^{-1/2} \left(\widetilde{\bV}_n  + \left(\frac{1}{\lambda} - 1\right)\bSigma^{-1}
         \right)^{-1}\bar\bx_n
                \bar\bx_n^\top \left(\widetilde{\bV}_n  + \left(\frac{1}{\lambda} - 1\right)\bSigma^{-1}
         \right)^{-1}\bSigma^{-1/2}\btheta
            }{
               \left(1 - \bar\bx_n^\top\bSigma^{1/2}  \left(\widetilde{\bS}_n + \left(\frac{1}{\lambda} - 1\right)\bI\right)^{-1} \bSigma^{1/2}\bar\bx_n\right)^2
            }.
    \end{align*}
}
Next, we prove that $\bar\bx_n^\top\left(\widetilde{\bV}_n  + \left(\frac{1}{\lambda} - 1\right)\bSigma^{-1} \right)^{-2}\bar\bx_n$ is bounded for $p/n\longrightarrow c \in (0,\infty)$ as $n\rightarrow\infty$.
{\footnotesize
    \begin{align}
  &      \bar\bx_n^\top\left(\widetilde{\bV}_n  + \left(\frac{1}{\lambda} - 1\right)\bSigma^{-1}
        \right)^{-2}\bar\bx_n \nonumber\\
&= \bar\bx_n^\top\bSigma^{1/2}\left(\tilde{\bS}_n  + \left(\frac{1}{\lambda} - 1\right)\bI
        \right)^{-1}\bSigma\left(\tilde{\bS}_n  + \left(\frac{1}{\lambda} - 1\right)\bI
        \right)^{-1}\bSigma^{1/2}\bar\bx_n   \nonumber\\
        & \leq  \lambda_{\max}(\bSigma)\cdot \bar\bx_n^\top\bSigma^{1/2}\left(\tilde{\bS}_n  + \left(\frac{1}{\lambda} - 1\right)\bI
        \right)^{-2}\bSigma^{1/2}\bar\bx_n\nonumber\\
& \leq \lambda_{\max}(\bSigma)\left(\frac{1}{\lambda} - 1\right)^{-1} \bar\bx_n^\top\bSigma^{1/2}\left(\widetilde{\bS}_n  + \left(\frac{1}{\lambda} - 1\right)\bI        \right)^{-1}\bSigma^{1/2}\bar\bx_n <\infty \label{eqn:bound}
    \end{align}
}
Using \eqref{eqn:bound} we get that
{\footnotesize
\begin{eqnarray*}
&&\left|\bxi^\top\bSigma^{-1/2} \left(\widetilde{\bV}_n  + \left(\frac{1}{\lambda} - 1\right)\bSigma^{-1}
         \right)^{-2}\bar\bx_n\right|\\
        &\le& \sqrt{\bxi^\top\bSigma^{-1/2} \left(\widetilde{\bV}_n  + \left(\frac{1}{\lambda} - 1\right)\bSigma^{-1}
         \right)^{-2}\bSigma^{-1/2}\bxi}\sqrt{\bar\bx_n^\top \left(\widetilde{\bV}_n  + \left(\frac{1}{\lambda} - 1\right)\bSigma^{-1}
         \right)^{-2}\bar\bx_n}<\infty \,.
\end{eqnarray*}
}
Hence, the application of \eqref{eqn:mean_stieltjes_conv_1}, \eqref{eqn:mean_stieltjes_conv_2}, and  Lemma \ref{lem2_lam} completes the proof of the lemma.
\end{proof}
\begin{proof}[Proof of Theorem \ref{th1-lam}:]
Let $V_{GMV}=1/(\ones^\top \bSigma^{-1} \ones)$. The application of the results of Lemma \ref{lem:lem3} with $\bxi=\bSigma\bb/\sqrt{\bb^\top \bSigma \bb}$ and $\btheta=\ones/\sqrt{\ones^\top \bSigma^{-1} \ones}$ leads to
\begin{align}
        \left|
            (L_{\bb}+1)^{-1/2}\bb^\top \bSigma \bS_{\lambda}^{-1}\ones  -  (L_{\bb}+1)^{-1/2}
            \bb^\top\bSigma\bOmega_{\lambda}^{-1} \ones
        \right| &\stackrel{a.s.}{\longrightarrow} 0,\label{pr_th1_eq1}\\
        \left|
            V_{GMV}\ones^\top \bS_{\lambda}^{-1}\ones  - \lambda^{-1} V_{GMV}
            \ones^\top\bOmega_{\lambda}^{-1} \ones
        \right| &\stackrel{a.s.}{\longrightarrow} 0,\label{pr_th1_eq2}\\
        \left| 
            V_{GMV}\ones^\top  \bS_{\lambda}^{-1} \bSigma \bS_{\lambda}^{-1} \ones -V_{GMV} (1-v_2^\prime(\eta, 0))\ones^\prime\bOmega_{\lambda}^{-1} \bSigma \bOmega_{\lambda}^{-1}\ones 
        \right| &\stackrel{a.s.}{\longrightarrow} 0\label{pr_th1_eq3}
    \end{align}

Using \eqref{pr_th1_eq1}-\eqref{pr_th1_eq3} and the equality
\begin{eqnarray*}
    L_{n;2}(\lambda)
        &=& \frac{\left(1-\frac{1}{\sqrt{L_{\bb}+1}} \frac{(L_{\bb}+1)^{-1/2}\bb^\top \bSigma \bS_{\lambda}^{-1}\ones}{V_{GMV}\ones^\top \bS_{\lambda}^{-1}\ones} \right)^2}
        {1-\frac{2}{\sqrt{L_{\bb}+1}}\frac{(L_{\bb}+1)^{-1/2}\bb^\top \bSigma \bS_{\lambda}^{-1}\ones}{V_{GMV}\ones^\top \bS_{\lambda}^{-1}\ones} +\frac{1}{L_{\bb}+1} \frac{V_{GMV}\ones^\top  \bS_{\lambda}^{-1} \bSigma \bS_{\lambda}^{-1} \ones}{\left(V_{GMV}\ones^\top \bS_{\lambda}^{-1}\ones\right)^2} } \\
\end{eqnarray*}
we get the statement of part (i) of the theorem, while the application of the equality
\begin{eqnarray*}
 \psi_n^*(\lambda) &=&\frac{\frac{1}{\sqrt{L_{\bb}+1}}\frac{(L_{\bb}+1)^{-1/2}\bb^\top \bSigma \bS_{\lambda}^{-1}\ones}{V_{GMV}\ones^\top \bS_{\lambda}^{-1}\ones} -\frac{1}{L_{\bb}+1} \frac{V_{GMV}\ones^\top  \bS_{\lambda}^{-1} \bSigma \bS_{\lambda}^{-1} \ones}{\left(V_{GMV}\ones^\top \bS_{\lambda}^{-1}\ones\right)^2}}
{1-\frac{2}{\sqrt{L_{\bb}+1}}\frac{(L_{\bb}+1)^{-1/2}\bb^\top \bSigma \bS_{\lambda}^{-1}\ones}{V_{GMV}\ones^\top \bS_{\lambda}^{-1}\ones} +\frac{1}{L_{\bb}+1} \frac{V_{GMV}\ones^\top  \bS_{\lambda}^{-1} \bSigma \bS_{\lambda}^{-1} \ones}{\left(V_{GMV}\ones^\top \bS_{\lambda}^{-1}\ones\right)^2} } 
\end{eqnarray*}
yields the second statement of the theorem.
\end{proof}

\begin{lem}\label{lem:trace_estimators}
Let $\frac{1}{p}\bSigma^{-1}$ possess a bounded trace norm. Then it holds that
\begin{flalign}
        &\left|\frac{1}{p}
           \tr\left(\left(\bS_n +\eta\bI\right)^{-1}\right) - \frac{1}{p}\tr\left(\left(v(\eta,0)\bSigma+\eta\bI\right)^{-1}\right)
        \right| \stackrel{a.s}{\rightarrow} 0 &\label{1_lem56}\\
        &\Bigg|
        \frac{\frac{1}{p}\tr\left((\bS_n +\eta \bI)^{-2}\right)\left(1-c+2c\eta\frac{1}{p}\tr\left((\bS_n +\eta \bI)^{-1}\right)\right)-c\left[\frac{1}{p}\tr\left((\bS_n +\eta \bI)^{-1}\right)\right]^2}{1-c+c\eta^2 \frac{1}{p}\tr\left((\bS_n +\eta \bI)^{-2}\right)}\nonumber\\
           &- \frac{1}{p}\tr\left(\left(v(\eta,0)\bSigma+\eta\bI  \right)^{-2}\right)
        \Bigg|  \stackrel{a.s}{\rightarrow} 0&\label{2_lem56}
    \end{flalign}
 for $p/n\rightarrow c \in (0,\infty)$ as $n\rightarrow\infty$.
\end{lem}

\begin{proof}[Proof of Lemma \ref{lem:trace_estimators}:]
    From part (i) of Lemma \ref{lem1_lam} with $\mathbf{\Theta}_p=\frac{1}{p}\bSigma^{-1}$ and the proof of Lemma \ref{lem:lem3} we obtain that $\frac{1}{p}\tr(\left(\bS_n +\eta\bI\right)^{-1})$ is consistent for $\frac{1}{p}\tr\left(\left(\eta\bI + v(\eta,0)\bSigma \right)^{-1}\right)$ in the high-dimensional setting.
 
Furthermore, applying part (ii) of Lemma \ref{lem1_lam}  with $\mathbf{\Theta}_p=\frac{1}{p}\bSigma^{-1}$ and following the proof of Lemma \ref{lem:lem3} we get that 
\begin{align}
&\frac{1}{p}\tr\left((\bS_n +\eta \bI)^{-2} \right) \stackrel{a.s.}{\rightarrow} 
        \frac{1}{p}\tr\left(\bSigma^{-1}
            \left(\eta\bSigma^{-1}+ v(\eta,0)\bI\right)^{-1}
\bSigma^{-1}\left(\eta\bSigma^{-1}+ v(\eta,0)\bI\right)^{-1}
        \right) \nonumber \\
&+ v_1^\prime(\eta,0) \frac{1}{p}\tr\left(\bSigma^{-1}
            \left(\eta\bSigma^{-1}+ v(\eta,0)\bI\right)^{-2}
        \right) \nonumber\\
& = \frac{1}{p}\tr\left(
\left(v(\eta,0)\bSigma+ \eta\bI\right)^{-2}\right) +   v_1^\prime(\eta,0) \frac{1}{p}\tr\left(\left(v(\eta,0)\bSigma+ \eta\bI\right)^{-1}\bSigma
\left(v(\eta,0)\bSigma+ \eta\bI\right)^{-1}
        \right)
            \nonumber
        \\
        = & \frac{v^\prime_1(\eta,0)}{v(\eta, 0)}\frac{1}{p}\tr\left(
            \left(v(\eta,0)\bSigma+\eta\bI\right)^{-1}\right) + \left(1- \frac{v^\prime_1(\eta,0)}{v(\eta, 0)}\eta \right)\frac{1}{p}\tr\left(\left( 
            v(\eta,0)\bSigma+\eta\bI\right)^{-2} \right), \label{eqn:lem.._eq1}
    \end{align}
where the application of \eqref{v1_eta_lam} leads to 
\begin{equation*}
\frac{v_1^\prime(\eta, 0)}{v(\eta,0)}  
= 
\frac{c\frac{1}{p}\text{tr}\left((v(\eta,0)\bSigma+\eta\bI)^{-1}\right)-c\eta \frac{1}{p}\text{tr}\left((v(\eta,0)\bSigma+\eta\bI)^{-2}\right)}
{1-c+2c\eta\frac{1}{p}\text{tr}\left((v(\eta,0)\bSigma+\eta\bI)^{-1}\right)-
c\eta^2\frac{1}{p}\text{tr}\left((v(\eta,0)\bSigma+\eta\bI)^{-2}\right)}.
\end{equation*}

Thus,
$\frac{1}{p}\tr\left((\bS_n +\eta \bI)^{-2}\right)$ converges almost surely to
\begin{eqnarray*}
&&
\frac{c\frac{1}{p}\text{tr}\left((v(\eta,0)\bSigma+\eta\bI)^{-1}
\right)-
c\eta\frac{1}{p}\text{tr}\left((v(\eta,0)\bSigma+\eta\bI)^{-2}
\right)}
{1-c+2c\eta\frac{1}{p}\text{tr}\left((v(\eta,0)\bSigma+\eta\bI)^{-1}
\right)-
c\eta^2\frac{1}{p}\text{tr}\left((v(\eta,0)\bSigma+\eta\bI)^{-2}
\right)}\\
&\times&\frac{1}{p}\tr\left(\left( 
            v(\eta,0)\bSigma+\eta\bI\right)^{-1} \right)\\
&+&\frac{1-c+c\eta\frac{1}{p}\text{tr}\left((v(\eta,0)\bSigma+\eta\bI)^{-1}\right)}
{1-c+2c\eta\frac{1}{p}\text{tr}\left((v(\eta,0)\bSigma+\eta\bI)^{-1}
\right)-
c\eta^2\frac{1}{p}\text{tr}\left((v(\eta,0)\bSigma+\eta\bI)^{-2}
\right)}\\
&\times&\frac{1}{p}\tr\left(\left( 
            v(\eta,0)\bSigma+\eta\bI\right)^{-2} \right) \\
&=& \frac{c\left[\frac{1}{p}\text{tr}\left((v(\eta,0)\bSigma+\eta\bI)^{-1}\right)\right]^2+(1-c)\frac{1}{p}\tr\left(\left( 
            v(\eta,0)\bSigma+\eta\bI\right)^{-2} \right) }
{1-c+2c\eta\frac{1}{p}\text{tr}\left((v(\eta,0)\bSigma+\eta\bI)^{-1}
\right)-
c\eta^2\frac{1}{p}\text{tr}\left((v(\eta,0)\bSigma+\eta\bI)^{-2}
\right)},      
\end{eqnarray*}
which together with \eqref{1_lem56} leads to the second statement of the lemma.
\end{proof}

\begin{proof}[Proof of Theorem \ref{th2-lam}:]
The result \eqref{lem:v_est-eq1} is a direct consequence of \eqref{v_eta_lam} in Lemma \ref{lem2_lam} and \eqref{1_lem56} in Lemma \ref{lem:trace_estimators}. 

Let $t_1=\frac{1}{p}\tr\left(\left(\bS_n+\eta\bI \right)^{-1}\right)$ and $t_2=\frac{1}{p}\tr\left(\left(\bS_n+\eta\bI \right)^{-2}\right)$. Then, the application of \eqref{v1_eta_lam} in  Lemma \ref{lem2_lam} and the results of Lemma \ref{lem:trace_estimators} leads to a consistent estimator of $v_1^\prime(\eta, 0)$ expressed as
\begin{eqnarray*}
\hat{v}_1^\prime(\eta, 0) &=&\hat{v}(\eta, 0)
\frac{c t_1-c\eta \frac{t_2(1-c+2c\eta t_1)-ct_1^2}{1-c+c\eta^2 t_2}}{1-c+2c\eta t_1-c\eta^2 \frac{t_2(1-c+2c\eta t_1)-ct_1^2}{1-c+c\eta^2 t_2}}\\
&=&\hat{v}(\eta, 0)\frac{(1-c)c(t_1-\eta t_2) +c^2\eta t_1(t_1-\eta t_2)}
{(1-c)^2+2(1-c)c \eta t_1+c^2\eta^2 t_1^2}=\hat{v}(\eta, 0)c\frac{(1-c+c\eta t_1)(t_1-\eta t_2)}{(1-c+c\eta t_1)^2}\\
&=&\hat{v}(\eta, 0)c\frac{t_1-\eta t_2}{(1-c+c\eta t_1)}=\hat{v}(\eta, 0)c\frac{t_1-\eta t_2}{\hat{v}(\eta,0)}.
\end{eqnarray*}

Finally, the result \eqref{lem:v_est-eq3} follows from \eqref{lem:v_est-eq1} and \eqref{lem:v_est-eq2} together with \eqref{v2_eta_lam} in Lemma \ref{lem2_lam}. 
\end{proof}

\begin{proof}[Proof of Theorem \ref{th3-lam}:]
The result \eqref{eq:lem58_0} is a special case of Theorem 3.2 in \cite{bodnar2014strong}. Equation \eqref{eq:lem58_1} follows from Lemma \ref{lem:lem3} with  $\bxi=\btheta=\ones/\sqrt{\ones^\top\bSigma^{-1}\ones}$. For the derivation of \eqref{eq:lem58_2} we note that
    \begin{align}
        \bb^\top \bSigma \bOmega_\lambda^{-1} \ones 
            = & 
            \frac{1}{\lambda v(\eta, 0)} \left(1 - (1-\lambda) \bb^\top \bOmega_\lambda^{-1} \ones\right)
    \end{align}
and apply Lemma \ref{lem:lem3} with  $\bxi=\bb/\sqrt{\bb^\top\bSigma^{-1}\bb} $ and $\btheta=\ones/\sqrt{\ones^\top\bSigma^{-1}\ones}$.

Finally, \eqref{eq:lem58_3} is obtained by noting that
    \begin{eqnarray*}
        \ones^\top\bOmega^{-1}_{\lambda}\bSigma\bOmega^{-1}_{\lambda}\ones = 
        \frac{1}{\lambda v(\eta, 0)} 
        \ones^\top\bOmega_{\lambda}^{-1}\ones
        - \frac{1-\lambda}{\lambda v(\eta, 0)}
        \ones^\top\bOmega_{\lambda}^{-2}\ones,
    \end{eqnarray*}
where \eqref{eq:lem58_1} is used for the first summand and \eqref{lem3_eq2} of Lemma \ref{lem:lem3} with $\bxi=\btheta=\ones/\sqrt{\ones^\top\bSigma^{-1}\ones}$ for the second one. 
\end{proof}

\bibliography{portfolio_refs.bib}

\end{document}